\documentclass[prx,twocolumn,english,superscriptaddress,floatfix,longbibliography]{revtex4-2}

\usepackage{natbib}
\usepackage{physics}
\usepackage{xcolor}
\usepackage{amsmath, amsfonts}
\usepackage{amsthm}
\usepackage{siunitx}
\usepackage{graphicx}
\usepackage[english]{babel}
\usepackage[utf8]{inputenc}
\usepackage[autostyle]{csquotes}
\MakeOuterQuote{"}
\usepackage[shortlabels]{enumitem}
\usepackage{dsfont}
\usepackage{comment}
\usepackage{verbatim}
\usepackage{xpatch}
\usepackage{float}
\usepackage{etoolbox}
\usepackage[normalem]{ulem} 
\newtheorem{theorem}{Theorem}
\newtheorem{corollary}{Corollary}[theorem]

\newtheorem{remark}{Remark}[theorem]
\usepackage[most]{tcolorbox}
\usepackage{booktabs}
\usepackage{tabularx}
\usepackage{subcaption}

\date{\today}

\definecolor{darkgreen}{rgb}{0.0, 0.8, 0.0}

\usepackage{etoolbox}
\usepackage{lipsum}

\newcommand{\deltaone}{\delta_1}

\newcommand{\epsATb}{\varepsilon_\mathrm{dep-1}}

\newcommand{\epssrc}{\varepsilon_\mathrm{ind}}

\newcommand{\epsspone}{\varepsilon_\mathrm{sp-1}}
\newcommand{\epssptwo}{\varepsilon_\mathrm{sp-2}}

\newcommand{\epsATc}
{\varepsilon_\mathrm{dep-2}}

\newcommand{\Sindep}
{S_{0,0}}

\newcommand{\Sdep}{S_{\delta_1,\delta_2}}

\newcommand{\Prindep}{\Pr_{\Sindep}}
\newcommand{\Prdep}{\Pr_{\Sdep}}

\newcommand{\Findep}{\mathcal{F}_{0,0}}

\newcommand{\Eindep}{\mathcal{E}_{0,0}}

\newcommand{\Eindepdecoy}{\Eindep^{\mathrm{decoy}}}

\newcommand{\Findepdecoy}{\Findep^{\mathrm{decoy}}}

\usepackage[hidelinks,colorlinks=true,linkcolor=blue,citecolor=blue]{hyperref}	
\usepackage[capitalize]{cleveref}

\crefname{appendix}{Appendix}{Appendices}

\makeatletter
\AddToHook{cmd/appendix/before}{\def\cref@section@alias{appendix}\def\cref@subsection@alias{appendix}}
\makeatother

\definecolor{pink}{RGB}{255,0,255}

\newcommand{\affvqcc}{Vigo Quantum Communication Center, University of Vigo, Vigo E-36310, Spain}
\newcommand{\affuvigo}{Escuela de Ingeniería de Telecomunicación, Department of Signal Theory and Communications, University of Vigo, Vigo E-36310, Spain}
\newcommand{\affatlantic}{atlanTTic Research Center, University of Vigo, Vigo E-36310, Spain}

\begin{document}
\author{Xoel Sixto}
    \email{xsixto@vqcc.uvigo.es}
	\affiliation{\affvqcc} \affiliation{\affuvigo} \affiliation{\affatlantic} 
\author{Guillermo Currás-Lorenzo}
        \affiliation{\affvqcc} \affiliation{\affuvigo} \affiliation{\affatlantic} 
\author{Margarida Pereira}
        \affiliation{\affvqcc} \affiliation{\affuvigo} \affiliation{\affatlantic}
\author{Víctor Zapatero}
	 	\affiliation{\affvqcc} \affiliation{\affuvigo} \affiliation{\affatlantic} 
\author{Álvaro Navarrete}
        \affiliation{\affvqcc} \affiliation{\affuvigo} \affiliation{\affatlantic} 
\author{Marcos Curty}
\affiliation{\affvqcc} \affiliation{\affuvigo} \affiliation{\affatlantic}

\title{Finite-key security analysis of decoy-state QKD with source and detector imperfections}
\begin{abstract}
Decoy-state quantum key distribution (QKD) is the most widely adopted approach for overcoming the limitations of imperfect single-photon sources. However, existing security proofs typically either neglect important device imperfections or rely on assumptions that are difficult to justify in realistic high-speed implementations, such as the independent and identically distributed nature of the emitted signals. In this work, we combine and extend several recent theoretical advances to provide a comprehensive analytical finite-key security proof for decoy-state QKD that simultaneously incorporates multiple practically relevant transmitter and receiver imperfections, including state-preparation flaws, bit and basis side-channel leakage and correlations, setting-independent intensity fluctuations, and detection-efficiency mismatches.
\end{abstract}

\maketitle

\section{Introduction}\label{sec:introduction}

Quantum key distribution (QKD) enables the establishment of symmetric cryptographic keys between two remote parties, conventionally referred to as Alice and Bob \cite{qkd1,qkd2,qkd3}. Unlike classical and post-quantum cryptographic schemes, whose security relies on assumptions regarding computational hardness, QKD guarantees security based on fundamental principles of quantum mechanics, such as the no-cloning theorem \cite{cloning}. Consequently, QKD can, in principle, provide information-theoretically secure communication even in the presence of an eavesdropper (Eve) equipped with unlimited computational capabilities.

Nevertheless, several challenges must still be addressed before QKD can be deployed on a large scale. One of the most critical issues is implementation security, namely bridging the gap between the assumptions underlying current security proofs and the actual behaviour of practical QKD devices \cite{zapatero2024implementation, BSI}. Even small discrepancies between theoretical assumptions and practical realizations may introduce security loopholes exploitable by Eve. For instance, many security analyses assume that Alice’s source emits single-photon states \cite{pereira_2020,Guille_framework}, an assumption that is difficult to fulfil with current technology. In practice, most QKD systems instead rely on laser sources emitting phase-randomized weak coherent pulses (PRWCPs). This mismatch between theoretical models and implementations is well known to enable photon-number-splitting (PNS) attacks \cite{pns3,pns2}, potentially compromising the security of the generated key.

The most efficient and widely adopted approach to overcome this limitation is the decoy-state method \cite{high_loss,decoy,decoy2}. Instead of using a single fixed intensity, Alice randomly modulates the intensity of the PRWCPs sent to Bob. This strategy enables a more accurate characterization of the quantum channel, allowing Alice and Bob to tightly estimate the yield and phase-error rate associated with the single-photon components from the measurement statistics obtained for different intensity settings. As a result, the decoy-state method achieves a secret key rate that scales linearly with the channel transmittance \cite{concise}, matching the performance attainable with ideal single-photon sources. Importantly, this technique also constitutes a key ingredient in other QKD protocols based on laser sources, such as measurement-device-independent (MDI) QKD \cite{MDI} and twin-field (TF) QKD \cite{twin}. Moreover, the decoy-state method has been extensively validated experimentally \cite{exp_decoy1,exp_decoy5,exp_decoy7,exp_decoy8,exp_decoy9,exp_decoy10,exp_decoy11} including in satellite-based QKD links \cite{sat1,sat5} and photonic integrated platforms \cite{exp_decoy12}. Indeed, this technique is currently used in several commercial QKD systems \cite{toshiba, thinkq, idquantique, telecom, qteck}.

Despite these advances, current security proofs for decoy-state QKD still fall short of establishing the security of practical implementations, as they neglect several important device imperfections present in realistic setups, including state-preparation flaws (SPFs) \cite{loss_tolerant, Mizutani_2015, pereira_2019,pereira_2020,Marcomini_2025}, side-channel leakage \cite{tha1,tha2,tha3,tha4,pereira_2019,pereira_2020,Sixto_2025}, and detection-efficiency mismatch \cite{ Tupkary2025phaseerrorrate,misma,Marcomini_2025}. Moreover, the proofs that do incorporate some of these imperfections typically consider restricted attack models or assume that Alice’s emitted signals are independent and identically distributed (IID). In practice, however, this assumption is difficult to satisfy due to the limited bandwidth of optical modulators and pulse correlations have been experimentally observed in several experiments \cite{Fadri_corre,Trefilov_corre,aless, agulleiro2025modelingcharacterizationarbitraryorder}. Although there are security proofs that address such correlations \cite{pereira_2020,Zapatero_2021,sixto2022,Currás-Lorenzo_2024}, no analytical proof for decoy-state QKD currently accounts for them together with side-channel leakage and encoding flaws.

In this work, we address this challenge by introducing a comprehensive analytical security proof for decoy-state QKD under realistic conditions that simultaneously incorporates the key hardware imperfections discussed above. To the best of our knowledge, this constitutes the first analytical framework capable of jointly addressing these practical vulnerabilities at both the transmitter and receiver sides. Specifically, our analysis accounts for SPFs, bit and basis side-channel leakage and encoding correlations, and bounded detector imperfections. To achieve this, it combines and extends several recent theoretical advances, building upon three main pillars.

$(i)$ {\it State-preparation flaws and side-channel leakage}. We build upon the security framework introduced in \cite{Guille_framework}. While the original formalism is restricted to single-photon protocols, we generalize it to make it applicable to decoy-state QKD.

$(ii)$ {\it Bit and basis correlations}. To address the breakdown of the IID assumption, we build upon the approach recently introduced in \cite{curraslorenzo2026rigorousphaseerrorestimationsecurityframework}. This technique decomposes the protocol into sub-protocols free from correlations, allowing the phase-error rate to be estimated independently for each of them, from which a bound for the global protocol can be derived. We adapt this decomposition strategy, originally developed for single-photon sources, to the specific structure of decoy-state QKD with bit and basis correlations.

$(iii)$  {\it Imperfect detectors}. Finally, we incorporate receiver-side vulnerabilities such as detection efficiency mismatch by integrating the mathematical tools introduced in \cite{misma} with the analysis above. 

This paper is organized as follows. In \cref{sec:protocol}, we describe the decoy-state protocol under study and introduce the main assumptions, including the characterization of the emitted states and the detector model. \cref{sec:security} presents our main theoretical result, namely a unified security proof for practical decoy-state QKD in the presence of the vulnerabilities discussed. In \cref{sec:simulations}, we evaluate the resulting secret key rate as a function of the magnitude of the different imperfections considered. Finally, \cref{sec:conclusions} summarizes our findings. For clarity and readability, detailed derivations and auxiliary proofs are deferred to the Appendices.

\section{Protocol description} \label{sec:protocol}

In each round $k\in \{1,2,...,N\}$ of the protocol, Alice  selects a basis $B_{A}\in\{Z,X\}$ with probabilities $p_{Z_A}$ and $p_{X_A} = 1- p_{Z_A}$ and a bit value uniformly at random, corresponding to one of the four possible BB84 states $j\in\{0_{Z},1_Z,0_X,1_X\}$. Also, she selects an intensity setting from the set $\{\mu_{1}, \mu_{2},...,\mu_{\mathcal{P}}\}$ with probabilities $\{p_{\mu_{1}}, p_{\mu_{2}},...,p_{\mu_{\mathcal{P}}}\}$ and sends to Bob a PRWCP encoded with the chosen settings. In reality however, the state generated might deviate from the given prescriptions due to device imperfections.

On the receiver side, Bob employs an active BB84 receiver to measure the incoming signals. Its action can be modelled as probabilistically choosing between two positive operator-valued measures (POVMs), namely $\left\{\Gamma_0^{({\xi})}, \Gamma_1^{({\xi})}, \Gamma_{\perp}^{({\xi})}\right\}$ with probability $p_{{\xi_B}}$ where $\xi\in\{Z,X\}$. Here, $\Gamma_b^{({\xi})}$ represents the POVM element corresponding to the bit value $b\in\{0,1\}$ and
$\Gamma_{\perp}^{(\mathcal{\xi})}$ corresponds to no detection.

\subsection{Generated states}\label{sec:generated}

We shall consider a scenario in which the state generated by Alice in a certain round $k$ depends not only on the setting choice $j_k$ but also on $all$ previous bit and basis setting choices $j_{k-1},...,j_{1}$. In addition, we assume that the emitted PRWCP suffers from setting-independent intensity fluctuations and there is no information leakage about the intensity setting selected. This means that the state in round $k$ can be written as 
\begin{equation}\label{eq:global_state}
\begin{aligned}
&\rho_{\mu_k, j_k \mid j_{k-1}, \ldots, j_1}=\\
&\sum_m p_{m| \mu_k}\left|\psi_{j_k \mid j_{k-1}, \ldots, j_1}^{m}\right\rangle\left\langle\psi_{j_k \mid j_{k-1}, \ldots, j_1}^{m}\right|_{T_k},
\end{aligned}
\end{equation}
where $m$ denotes the photon number, the subscript $T_k$ identifies the state sent to the channel and 
\begin{equation}
\label{eq:pns}
p_{m|\mu_{k}}=\int_{\mu^{L}_k}^{\mu^{U}_k}g_{\mu_k}(\alpha)\frac{e^{-\alpha}\alpha^{m}}{m!}d\alpha.
\end{equation}
In this last equation we consider that the quantities $\mu^{L(U)}_k$ and the probability density function describing the intensity distribution $g_{\mu_k}(\alpha)$ are known, so the above statistics can be computed exactly. Furthermore, for the security proof to apply, the three conditions below should hold. 

{\it (i)} The single-photon components emitted by Alice are $\epsilon_{\rm side}$-close  to some characterized reference states $\{\ket{\phi_{j_k}}_{T_k}\}$, that is
\begin{equation}
\label{eq:close}
\left|\left\langle\phi_{j_k} \Big| \psi_{j_k \mid j_{k-1}=\gamma, j_{k-2}=\gamma, \ldots, j_{1}=\gamma}^{1}\right\rangle_{T_k}\right|^2 \geq 1-\epsilon_{\text {side }},
\end{equation}
where all the previous settings are fixed to a certain $\gamma\in\{0_Z,1_Z,0_X,1_X\}$. The parameter $\epsilon_{\rm side}$ accounts for information leakage about the bit and basis selection. As for the reference states we take, for simplicity, single-photon qubit states of the form
\begin{equation}
\left|\phi_{j_k}\right\rangle_{T_k}=\cos \left(\theta_{j_k}\right)\left|0_Z\right\rangle_{T_k}+\sin \left(\theta_{j_k}\right)\left|1_Z\right\rangle_{T_k},
\end{equation}
where $\theta_{j_k}=\left(1+\delta_{\mathrm{SPF}} / \pi\right) \varphi_{j_k} / 2$, the term characterizing qubit SPFs satisfies $\delta_{\mathrm{SPF}} \in[0, \pi)$ and the parameter $\varphi_{j_k} \in\{0, \pi, \pi / 2,3 \pi / 2\}$ for $j_{k} \in\left\{0_Z, 1_Z, 0_X, 1_X\right\}$. We note, however, that our analysis could be readily adapted to any other characterized SPF model.

Importantly, in the absence of correlations \cref{eq:close} reduces to 
\begin{equation}
\left|\left\langle\phi_{j_k} \Big| \psi_{j_k}^{1}\right\rangle_{T_k}\right|^2 \geq 1-\epsilon_{\text{side}}.
\end{equation}
In the presence of bit and basis correlations, on the other hand, two additional conditions are needed:

{\it (ii)} Let the state 
\begin{equation}
\begin{aligned}
\label{eq:source-replaced_state}
&\left|\Psi_{j_k \mid j_{k-1}, \ldots, j_1}\right\rangle_{I_k M_k T_k}=\\
&\sum_{\mu_k} \sqrt{p_{\mu_k}}\left|{\mu}_k\right\rangle_{I_k} \sum_{m=0}^{\infty} \sqrt{p_{m |\mu_k}}|m\rangle_{M_k}\left|\psi_{j_k \mid j_{k-1}, \ldots, j_1}^{m}\right\rangle_{T_k},
\end{aligned}
\end{equation}
represent a purification of the state described in \cref{eq:global_state} including Alice's intensity setting choice in round $k$. Precisely, $I_k$ and $M_k$ represent the intensity and photon number registers that are kept in Alice's lab. These states must satisfy
\begin{equation}
\label{eq:corre_WCP}
\begin{aligned}
&\left|\left\langle\Psi_{j_k \mid j_{k-1}, \ldots, \tilde{j}_{k-l}, \ldots, j_1} \Big| \Psi_{j_k \mid j_{k-1}, \ldots, j_{k-l}, \ldots, j_1}\right\rangle_{I_k M_k T_k}\right|^2\\
&\geq 1-\epsilon^{\rm WCP}_l, 
\end{aligned}
\end{equation}
for all rounds  $k\in\{1,2,...,N\}$ and $l\in\{1,2,...,k-1\}$, and every pair of bit and basis setting choices in round $k-l$, $\tilde{j}_{k-l}$ and $j_{k-l}$. In \cref{eq:corre_WCP}, $\epsilon_{l}^{\rm WCP}$ denotes the correlation strength with pulse separation $l$.

{\it (iii)} The single-photon components emitted by Alice satisfy
\begin{equation}
\label{eq:close_2}
\left|\left\langle\psi_{j_k \mid j_{k-1},\ldots, \tilde{j}_{k-l},\ldots, j_{1}}^{1} \Big| \psi_{j_k \mid j_{k-1},\ldots, {j}_{k-l}, \ldots, j_{1}}^{1}\right\rangle_{T_k}\right|^2 \geq 1-\epsilon_{l}^{\text {SP }},
\end{equation}
for every $k\in\{1,2,...,N\}$, $l\in\{1,2,...,k-1\}$ and every pair of settings $\tilde{j}_{k-l}$ and ${j}_{k-l}$.
Similarly to condition $(ii)$, here $\epsilon_{l}^{\text {SP }}$ represents the correlation strength associated with the pulse separation $l$, but now for the single-photon component.

In general, one expects that the parameters $\epsilon^{\Xi}_l$, with $\Xi\in\{\rm WCP, SP\}$, are a monotonically decreasing function of $l$ satisfying $\epsilon^{\Xi}_l \to 0$ as $l \to \infty$ \cite{pereira2024}. Based on the experimental results in \cite{agulleiro2025modelingcharacterizationarbitraryorder}, for concreteness, we shall assume an exponential decay of the form
\begin{equation}
\label{eq:corre_model}
\epsilon^{\Xi}_l = \epsilon_1^{\Xi} e^{-C^{\Xi}(l-1)} ,
\end{equation}
where $\epsilon_1^{\Xi}$ denotes the magnitude of nearest-neighbor pulse correlations and $C^{\Xi}$ is a known decay constant. Physically, this model implies that correlations weaken as the separation between rounds increases, vanishing in the limit of infinitely separated rounds. In the following, for simplicity, we shall model $\epsilon_{l}^{\rm WCP}$ and $\epsilon_{l}^{\rm SP}$ according to the same correlation strength, $i.e.,$  $\epsilon_1^{\Xi}\equiv\epsilon_1$ and $C^{\Xi}\equiv C$ $\forall \Xi\in\{{\rm WCP, SP}\}$  which corresponds to a worst-case scenario where all correlations present in the PRWCP state are attributed to the single-photon component. This is a conservative assumption because bit and basis correlations could in principle be distributed among the vacuum, single-photon, and multiphoton components.  Therefore, by assigning the entire correlation effect to the single-photon component, which is the one used for key generation we effectively minimize the secret key rate.

\subsection{Detector model}\label{sec:detector}

For the detector, we adopt the active model proposed in \cite{Tupkary2025phaseerrorrate} in which the dark count rate $d_{\xi_b}$ and the detector efficiency $\eta_{\xi_b}$ of the detector associated with basis $\xi$ and the bit value $b$ are characterized only within finite tolerances, $\Delta_{\rm dc}$ and $\Delta_{\eta}$, as follows
\begin{equation}
\begin{aligned}
& d_{\xi_b} \in\left[d_{\operatorname{dc}}\left(1-\Delta_{\mathrm{dc}}\right), d_{\operatorname{dc}}\left(1+\Delta_{\mathrm{dc}}\right)\right],\\
& \eta_{\xi_b} \in\left[\eta_{\operatorname{det}}\left(1-\Delta_\eta\right), \eta_{\operatorname{det}}\left(1+\Delta_\eta\right)\right].
\end{aligned}
\end{equation}

For this, we introduce two parameters, $\delta_1$ and $\delta_2$ that parametrize the detection efficiency mismatch and quantify the departure from the basis-independent detection efficiency condition.
They can be bounded as \cite{Tupkary2025phaseerrorrate}
\begin{equation}
\begin{aligned}
\delta_1 \leq \max & \left\{\left(1-\frac{1-\left(1-d_{\min }\right)^2}{1-\left(1-d_{\max }\right)^2}\right) \frac{d_{\max }\left(2-d_{\min }\right)}{1-\left(1-d_{\min }\right)^2}\right., \\
& \left.4\left|1-\sqrt{1-\left(1-d_{\min }\right)^2\left(1-r_\eta\right)}\right|\right\},
\end{aligned}
\end{equation}
and
\begin{equation}
\delta_2 \leq \max \left\{1-\frac{1-\left(1-d_{\min }\right)^2}{1-\left(1-d_{\max }\right)^2},\left(1-d_{\min }\right)^2\left(1-r_\eta\right)\right\},
\end{equation}
where
\begin{equation}
\begin{aligned}
& d_{\max }=\max _{\xi \in\{X, Z\}}\left\{d_{\xi_0}, d_{\xi_1}\right\} \leq d_{\mathrm{dc}}\left(1+\Delta_{\mathrm{dc}}\right), \\
& d_{\min }=\min _{\xi \in\{X, Z\}}\left\{d_{\xi_0}, d_{\xi_1}\right\} \geq d_{\mathrm{dc}}\left(1-\Delta_{\mathrm{dc}}\right), \\
& r_\eta=\left(1-\Delta_\eta\right) /\left(1+\Delta_\eta\right) .
\end{aligned}
\end{equation}
We note that other detector models could be incorporated into our security proof as well, provided that bounds on $\delta_1$ and $\delta_2$ can be established. For a more detailed discussion of these parameters, we refer the reader to \cite{Tupkary2025phaseerrorrate,misma}. 

\section{Security proof}\label{sec:security}
Here we provide a  description of our security proof for a finite-key implementation of the decoy-state QKD protocol introduced in the previous section. It follows the variable key-length approach introduced in \cite{Tupkary2025phaseerrorrate} based on entropic uncertainty relations (EUR) and the leftover hashing lemma (LHL) framework \cite{EUR_1,EUR_2,EUR_3,leftover}. Precisely, it can be shown (see Appendix~\ref{app:rigurous}) that a lower bound on the length of the final key $\ell$ is given by 
\begin{equation}
\begin{aligned}
\label{eq:key_l}
\ell \geq&\Big\lfloor\underline{N}_{Z,1}^{\mathrm{det},Z} \left(1-h\left(\overline{e}_{\mathrm{ph},1}\right)\right)-\lambda_{\mathrm{EC}}\\
&-2 \log_2 \left(1 /(2 \varepsilon_{\mathrm{PA}})\right)-\log_2 \left(2 / \varepsilon_{\mathrm{corr}}\right)\Big\rfloor,
\end{aligned}
\end{equation}
being the overall security parameter $\varepsilon_{\mathrm{corr}}+\varepsilon_{\mathrm{sec}}$, where
$\varepsilon_{\mathrm{corr}} $ denotes the failure probability of the error-verification step, and
$\varepsilon_{\mathrm{sec}} = 2\sqrt{\varepsilon} + \varepsilon_{\mathrm{PA}}$, with $\varepsilon$ being the failure probability of the parameter estimation step and $\varepsilon_{\mathrm{PA}}>0$ a freely chosen parameter. In \cref{eq:key_l}, $\underline{N}_{Z,1}^{{\rm det},Z}$ denotes a lower bound on the number of detected single-photon rounds in which both Alice and Bob select the $Z$ basis (the key-generating rounds),  $h(\cdot)$ is the binary entropy function, $\overline{e}_{\mathrm{ph},1}$ represents an upper bound on the phase-error rate, $\lambda_{\rm EC}$ denotes the amount of bits revealed during error correction, and $\log_2(2/\varepsilon_{\mathrm{corr}})$ corresponds to the number of bits disclosed during error verification \footnote{From now on, lower scripts in number of counts like in $N_{Z,1}^{{\rm det},Z}$ denote Alice selection and the super script denotes Bob's selection, while underlines represent lower bounds and overlines denote upper bounds.}.

To lower bound the number of single-photon detected rounds with $Z$-basis match, $ N_{Z,1}^{{\rm det},Z}$, and upper bound the single-photon phase-error rate, $e_{\mathrm{ph},1}$, in the presence of correlations, we extend the approach introduced in \cite{curraslorenzo2026rigorousphaseerrorestimationsecurityframework} to the decoy-state setting (see \cref{app:rigurous}). Specifically, we partition the protocol data into $l_c+1$ groups, where $l_c$ denotes the correlation length, and estimate the relevant statistical bounds independently for each group. These bounds are then combined in a single bound for each relevant quantity to compute the secret key length following \cref{eq:key_l}.

When the bit and basis correlations are unbounded, $i.e.,$ they are of infinite length, we define an effective maximum correlation length as \cite{pereira2024} 
\begin{equation}
    l_c = \left\lceil{\frac{1}{C}\operatorname{ln}\left(\frac{N\epsilon_1}{\zeta^2(1-\sqrt{e^{-C}})^2}\right)}\right\rceil,
\end{equation}
where $C$ is defined in \cref{eq:corre_model} and $\zeta>0$ represents a parameter that can be freely chosen. In this latter case the secrecy parameter takes the form $\varepsilon_{\rm sec}=2\sqrt{\varepsilon+\zeta}+\varepsilon_{\rm PA}$ \cite{pereira2024}. 

Next, we estimate the required bounds on  $ N_{Z,1}^{{\rm det},Z}$ and  $e_{\mathrm{ph},1}$.
 
\subsection{Lower bound on $N_{Z,1}^{{\rm det},Z}$}\label{sec:decoy}

We begin by estimating $\underline{N}_{Z,1}^{{\rm det},Z}$, in any group $w\in \{0,1,...,l_c\}$. For notational simplicity, however, we shall omit the explicit dependence of this parameter on the subgroup index $w$, with the understanding that all estimations below correspond to an individual group. 

We follow the analysis in \cite{Mannalath_decoy} (see also \cite{concise,Attema_2021}) with slight modifications. Precisely, we consider the standard counterfactual scenario in which Alice is assumed to send photon-number states to Bob and the corresponding intensity settings are assigned to the detected events \emph{a posteriori}. 

The expected value of $N^{{\rm det} , Z}_{Z, \mu_p}$, $i.e.$, the number of detected rounds in which both Alice and Bob select the $Z$ basis and the intensity of Alice's emitted pulses is $\mu_p$ reads
\begin{equation}\label{eq:expected}
\mathbb{E}\left[N^{{\rm det},Z}_{Z, \mu_p}\right]=\sum_{m=0}^{\infty} p_{\mu_p \mid m} N^{{\rm det},Z}_{Z, m},
\end{equation}
where $N^{{\rm det},Z}_{Z, m}$ denotes the number of detected events in which Alice and Bob choose the $Z$ basis and Alice emits an $m$-photon pulse. Due to Bayes' rule we have that
\begin{equation}
p_{\mu_p \mid m}=\frac{p_{m \mid \mu_p} p_{\mu_p}}{p_{m}},
\end{equation}
with $p_{m}=\sum_{p=1}^{\mathcal{P}} p_{m \mid \mu_p} p_{\mu_p }$ and where $\mathcal{P}$ denotes the number of intensities.

In the counterfactual scenario, the observed variables $N^{{\rm det},Z}_{Z, \mu_p}$ arise from independent Bernoulli samples of $N^{{\rm det},Z}_{Z}$ $i.e.$,  the detected rounds in which Alice and Bob select the $Z$ basis. Hence, using \cref{th1} from Appendix~\ref{sec:concentration} we obtain that
\begin{equation}
\mathbb{E}\left[N^{{\rm det},Z}_{Z, \mu_p}\right]=N^{{\rm det},Z}_{Z, \mu_p}+\delta^{ Z}_{\mu_p, N^{{\rm det},Z}_Z},
\end{equation}
where $ \delta_{\mu_p, N^{{\rm det},Z}_Z}^{Z,\mathrm{L}} \leq \delta^{Z}_{\mu_p, N^{{\rm det},Z}_{Z}} \leq \delta_{\mu_p, N^{{\rm det},Z}_{Z}}^{Z,\mathrm{U}},$ except with probability $\varepsilon_{\mu_p, N^{{\rm det},Z}_Z}^{\mathrm{L}}+\varepsilon_{\mu_p, N^{{\rm det},Z}_Z}^{\mathrm{U}}$ with
\begin{equation}
\begin{aligned}
\delta_{\mu_p, N^{{\rm det},Z}_{Z}}^{Z,\Delta}=&-N^{{\rm det},Z}_{Z, \mu_p}+\\
&N^{{\rm det},Z}_{Z} \mathcal{F}_{N^{{\rm det},Z}_{Z}, \varepsilon_{\mu_p, N^{{\rm det},Z}_{Z}}^{\Delta}}^{\Delta}\left(N^{{\rm det},Z}_{Z, \mu_p} / N^{{\rm det},Z}_{Z}\right),
\end{aligned}
\end{equation}
for $\Delta\in\{{\rm L, U}\}$ and where the function $\mathcal{F}^{\Delta}_{N, \varepsilon}$ is defined in \cref{eq:F_def}. Furthermore, since $\sum_{p=1}^\mathcal{P}\mathbb{E}\left[N^{{\rm det},Z}_{Z, \mu_p}\right]=N^{{\rm det},Z}_{Z}$
it holds that $\sum_{p=1}^{\mathcal{P}} \delta^{Z}_{\mu_p, N^{{\rm det},Z}_Z}=0$.

To estimate $\underline{N}_{Z, 1}^{{\rm det},Z}$ using linear programming, we need to reduce the number of unknowns in \cref{eq:expected} to a finite set. For this, we  define a threshold photon number $m_{\rm max}$. This results in two separate bounds, namely
\begin{equation}
N^{{\rm det},Z}_{Z, \mu_p}+\delta^{ Z}_{\mu_p, N^{{\rm det},Z}_{Z}} \geq \sum_{m=0}^{m_{\rm max}} p_{\mu_p \mid m} N^{{\rm det},Z}_{Z, m},
\end{equation}
and using \cref{th2} from Appendix~\ref{sec:concentration} for the complementary bound, it follows that
\begin{equation}
\begin{aligned}
\label{eq:20}
&N^{{\rm det},Z}_{Z, \mu_p}+\delta^{Z}_{\mu_p, N^{{\rm det},Z}_{Z}}\leq\\
&\sum_{m=0}^{m_{\rm max}} p_{\mu_p \mid m} N^{{\rm det},Z}_{Z, m}+N^{Z}_{Z}\left(1-\sum_{m=0}^{m_{\rm max}} p_{m}\right)+\delta_{Z,> m_{\rm max}}.
\end{aligned}
\end{equation}

In \cref{eq:20}, the parameter $N_Z^Z$ represents the number of pulses in which Alice and Bob select the $Z$ basis and $\delta_{Z,> m_{\rm max}}^{\mathrm{L}} \leq \delta_{Z,> m_{\rm max}} \leq \delta_{Z,> m_{\rm max}}^{\mathrm{U}}$ with
\begin{equation}
    \delta_{Z, m}^{\Delta}=N^{Z}_{Z} \mathcal{G}_{N^{Z}_{Z}, \varepsilon_{Z, m}^{\Delta}}^{\Delta}\left(p_{m}\right)-p_{m} N^{Z}_{Z} ,
\end{equation}
which holds except with probability $\varepsilon_{Z, >m_{\rm max}}^{\mathrm{L}}+\varepsilon_{Z, >m_{\rm max}}^{\mathrm{U}}$. The function $\mathcal{G}_{N, \varepsilon}^{\Delta}$ in given by \cref{eq:G0,eq:G1}.

Also, by noting that $N^{{\rm det},Z}_{Z, m}$ is upper bounded by the number of $m$-photon pulses sent and measured in the $Z$ basis, $N_{Z, m}^{Z}$, and employing \cref{th2} from Appendix~\ref{sec:concentration} it follows that

\begin{equation}
    0 \leq N^{{\rm det}, Z}_{Z, m} \leq \min \left(p_{m} N_{Z}^{Z}+\delta _{Z, m}, N^{{\rm det}, Z}_{Z}\right),
\end{equation}
where $\delta_{Z, m}^{\mathrm{L}} \leq \delta_{Z, m} \leq \delta_{Z, m}^{\mathrm{U}}$ except with probability $\varepsilon_{Z, m}^{\mathrm{L}}+\varepsilon_{Z, m}^{\mathrm{U}}$.

Finally, as 
\begin{equation}
\sum_{m=0}^{m_{\rm max}}N^{Z}_{Z, m}=\sum_{m=0}^{m_{\rm max}}p_m N^{Z}_Z+\sum_{m=0}^{m_{\rm max}}\delta_{Z, m},
\end{equation}
we have that $\delta_{Z,\leq m_{\rm max}}^{\mathrm{L}} \leq \sum_{m=0}^{m_{\rm max}} \delta_{Z, m} \leq \delta_{Z,\leq m_{\rm max}}^{\mathrm{U}}$ except with probability 
\begin{equation}
 \sum_{m=0}^{m_{\rm max}}(\varepsilon_{Z,m}^{\rm L}+\varepsilon_{Z,m}^{\rm U})+\varepsilon_{Z, \leq m_{\rm max}}^{\mathrm{L}}+\varepsilon_{Z, \leq m_{\rm max}}^{\mathrm{U}}.
\end{equation}
By taking the infinite dimensional sum, we obtain the following constraint 
\begin{equation}
\sum_{m=0}^{m_{\rm max}} \delta_{Z, m}+\delta_{Z,> m_{\rm max}}=0.    
\end{equation}

Now, we have all the necessary constraints to construct 
a linear program (LP) to estimate $\underline{N}_{Z,1}^{{\rm det}, Z}$ that holds except with probability at most
\begin{equation}
\begin{aligned}
&\varepsilon_{Z,1}=\\
&\sum_{p=1}^\mathcal{P}\left(\varepsilon_{\mu_p, N^{{\rm det},Z}_Z}^{\mathrm{L}}+\varepsilon_{\mu_p, N^{{\rm det},Z}_Z}^{\mathrm{U}}\right)+\sum_{m=0}^{m_{\rm max}}\left(\varepsilon_{Z, m}^{\mathrm{L}}+\varepsilon_{Z, m}^{\mathrm{U}}\right)\\
&+\varepsilon_{Z, \leq m_{\rm max}}^{\mathrm{L}}+\varepsilon_{Z, \leq m_{\rm max}}^{\mathrm{U}}+\varepsilon_{Z, >m_{\rm max}}^{\mathrm{L}}+\varepsilon_{Z, >m_{\rm max}}^{\mathrm{U}}.
\end{aligned}
\end{equation}
The LP is given by 
\begin{widetext}
\begin{equation}
\begin{aligned} \label{eq:LP_key}
\underline{N}_{Z,1}^{{\rm det},Z}=\text { min } & N^{{\rm det},Z}_{Z,1} \\
\text { s.t. } & N^{{\rm det},Z}_{Z, \mu_p}+\delta^{Z}_{\mu_p, N^{{\rm det},Z}_Z} \geq \sum_{m=0}^{m_{\rm max}} p_{\mu_p \mid m} N^{{\rm det},Z}_{Z, m} \\
& N^{{\rm det},Z}_{Z, \mu_p}+\delta^{Z}_{\mu_p, N^{{\rm det},Z}_Z} \leq \sum_{m=0}^{m_{\rm max}} p_{\mu_p \mid m} N^{{\rm det},Z}_{Z, m}+N^{Z}_Z\left(1-\sum_{m=0}^{m_{\rm max}} p_{m}\right)+\delta_{Z, >m_{\rm max}} \\
& 0 \leq N^{{\rm det},Z}_{Z, m} \leq \min \left(p_{m} N^{Z}_Z+\delta_{Z, m}, N^{{\rm det},Z}_Z\right) \\
& \delta_{\mu_p, N^{{\rm det},Z}_Z}^{Z,\mathrm{L}} \leq \delta^{ Z}_{\mu_p, N^{{\rm det},Z}_Z} \leq \delta_{\mu_p, N^{{\rm det},Z}_Z}^{Z,\mathrm{U}} \\
& \delta_{Z, >m_{\rm max}}^{\mathrm{L}} \leq \delta_{Z, >m_{\rm max}} \leq \delta_{Z, >m_{\rm max}}^{\mathrm{U}} \\
& \delta_{Z, m}^{\mathrm{L}} \leq \delta_{Z, m} \leq \delta_{Z, m}^{\mathrm{U}} \\
& \sum_{p=1}^\mathcal{P} \delta^{Z}_{\mu_p, N^{{\rm det},Z}_Z}=0 \\
& \delta_{Z,\leq m_{\rm max}}^{ \mathrm{L}} \leq \sum_{m=0}^{m_{\rm max}} \delta_{Z, m} \leq \delta_{Z, \leq m_{\rm max}}^{\mathrm{U}} \\
& \sum_{m=0}^{m_{\rm max}} \delta_{Z, m}+\delta_{Z, >m_{\rm max}}=0 \\
& \forall~1 \leq p \leq \mathcal{P}, \forall~0 \leq m \leq m_{\rm max}.
\end{aligned}
\end{equation}
\end{widetext}
Once $\underline{N}_{Z,1}^{\mathrm{det}, Z}$ has been estimated for all groups $w$, which we shall denote by $\underline{N}_{Z,1}^{\mathrm{det}, Z, w}$ for each group $w$, the overall lower bound on the parameter ${N}_{Z,1}^{\mathrm{det}, Z}$ appearing in \cref{eq:key_l} is given by 
\begin{equation}
\label{eq:sum_key_rounds}
\underline{N}_{Z,1}^{\mathrm{det},Z}= \sum_w \underline{N}_{Z,1}^{\mathrm{det},Z,w}.   
\end{equation}

\subsection{Upper bound on $e_{{\rm ph},1}$
}\label{sec:phase}

To estimate an upper bound on the single-photon phase-error rate, we build upon the security framework introduced in~\cite{Guille_framework} for uncorrelated single-photon protocols and we refer the reader to that work for the detailed derivations. The analysis introduces a fictitious quantum coin $C$ ~\cite{GLLP} that Alice randomly measures in either the $Z$ or $X$ basis. When the coin is measured in the $Z$ basis, the two possible outcomes project the joint state shared between Alice and Bob onto two carefully constructed states with two possible tags $t\in \{{\rm TAR, REF}\}$. The statistics of the target state, $|\Psi_{\rm TAR}\rangle$,
are related to the phase-error rate, while those of the reference state, $|\Psi_{\rm REF}\rangle$, are determined using the data acquired in the actual protocol. Importantly, the quantum coin inequality  bounds the phase-error contribution in terms of the  reference state statistics and the number of times Alice obtains the result $X_C=1$ when measuring the coin in the $X$ basis. This latter quantity reflects the imbalance of the coin and depends on how much the actual emitted states deviate from the ideal ones. Finally, a  fictitious tag-assignment step bridges the coin-round analysis with the actual protocol rounds, allowing the phase-error bound to be extended to all rounds via random sampling arguments \cite{random_sampling}.

Extending this framework to the present setting requires incorporating bit and basis correlations, a decoy-state analysis and detector imperfections.
To incorporate the former, we use the same protocol partitioning strategy introduced before and evaluate the phase-error rate separately within each resulting uncorrelated group. This can be done following the observation that the subset of single-photon emissions forms an effective subprotocol that is indistinguishable from one in which Alice employs an ideal single-photon source and to evaluate the relevant statistics, we employ the decoy-state technique. Ultimately, we incorporate detector imperfections by suitably adapting the results of~\cite{misma}, as detailed in Appendix~\ref{app:detector}. 

The procedure to derive an upper bound on the phase-error rate for each group $w$ is detailed below; however, for simplicity of notation, we again omit the explicit dependence on the group $w$. It is worth emphasizing that our ultimate goal is to estimate the phase-error rate of the full protocol, rather than that of each individual group. For this, in Appendix~\ref{app:rigurous}, we show how to relate the global phase-error rate to those of the $l_c+1$ groups by extending the results of \cite{curraslorenzo2026rigorousphaseerrorestimationsecurityframework}.

In particular, as shown in \cite{misma}, an upper bound on the number of single-photon phase errors in group $w$ is given by
\begin{equation}
\begin{aligned} \label{eq:number_errors}
\overline{N}_{\mathrm{ph},1}=\frac{1}{p_{Z_C} p_{\mathrm{TAR} \mid \mathrm{vir}}}\Big[&\left(\overline{N}_{Z_C=0,1}^{\mathrm{det}}+\Delta_A\right) G_{+}(y, z)\\
&-\underline{N}_{0_Z, \mathrm{TAR}, Z_C,1}^{0_X}+\Delta_A+\Delta_H\Big].
\end{aligned}
\end{equation}
Here  $p_{Z_C}$ represents the probability of measuring the coin system in the $Z$ basis, which can be freely chosen, and
\begin{equation}
    p_{\mathrm{TAR} \mid \mathrm{vir}}=\frac{p_{\mathrm{TAR}}}{p_{Z_A} p_{Z_B}\left(1+c_1 q_{\mathrm{vir} 1}\right)} 
\end{equation}
with
\begin{equation}\label{eq:ptar}
\begin{aligned}
p_{\mathrm{TAR}}=\left(1+c_1 q_{\mathrm{vir} 1}\right) \min \Bigg\{p_{Z_A} p_{Z_B}, \frac{p_{Z_A} p_{X_B}}{2 c_1 q_{\mathrm{vir} 1}}, \frac{p_{Z_A} p_{X_B}}{2 c_2 q_{\mathrm{vir} 1}}, \\
\frac{p_{X_A} p_{X_B}}{2\left(1-q_{\mathrm{vir} 1}\right)}, \frac{p_{X_A} p_{X_B}}{2 c_3 q_{\mathrm{vir} 1}}\Bigg\},
\end{aligned}
\end{equation}
where the quantity $q_{\mathrm{vir} 1}$ is given by
\begin{equation}\label{eq:qvir}
q_{\mathrm{vir} 1}=\frac{1}{2}(1-\text{cos}(\kappa \pi/2))
\end{equation}
with $\kappa=1+\delta_{\rm SPF}/\pi$, and
\begin{equation}\label{eq:c_defs}
\begin{aligned}
c_1 & =\frac{\cos (\kappa \pi / 2)}{\cos (\kappa \pi)-\cos (\kappa \pi / 2)}, \\
c_2 & =\frac{\cos (\kappa \pi / 2)}{\cos (\kappa \pi / 2)-1}, \\
c_3 & =\frac{1+\cos (\kappa \pi / 2)}{\cos (\kappa \pi / 2)-\cos (\kappa \pi)} .
\end{aligned}
\end{equation}
Furthermore, in \cref{eq:number_errors} $\overline{N}_{Z_C=0,1}^{\mathrm{det}}$ denotes an upper bound on the number of detected single-photon rounds in which $Z_C=0$ and the terms
$\Delta_A$ and $\Delta_H$ arise from applying Azuma's \cite{Azuma} and Hoeffding's \cite{Hoeff} inequality respectively (see Appendix~\ref{sec:concentration}). They are
defined as 
\begin{equation}
\begin{aligned}
\Delta_A & =\sqrt{2 \overline{N}_{1}^{\mathrm{det}} \ln \left(1 / \varepsilon_{\mathrm{conc}}\right)}, \\
\Delta_H & =\sqrt{\frac{\hat{N}_{Z,1}^{\mathrm{det},Z} \ln \left(1 / \varepsilon_{\mathrm{conc}}\right)}{2}} ,
\end{aligned}
\end{equation}
with $\overline{N}_{1}^{\mathrm{det}}:=\overline{N}_{Z,1}^{\mathrm{det},Z}+\overline{N}_{1}^{\mathrm{det},X}$, where ${N}_{1}^{\mathrm{det},X}$ represents the number of detected single-photon rounds measured in the $X$ basis  \footnote{Here, we require an upper bound on the number of detected single-photon signals. For that we consider that ${N}_{1}^{\mathrm{det}}:={N}_{Z,1}^{\mathrm{det},Z}+{N}_{1}^{\mathrm{det},X}$. An upper bound on ${N}_{Z,1}^{\mathrm{det},Z}$ can be obtained with the decoy-state analysis in \cref{sec:decoy} while to obtain an upper bound on ${N}_{1}^{\mathrm{det},X}$ we consider that
$\overline{N}_{1}^{\mathrm{det},X}:= \overline{N}_{0_Z,1}^{\mathrm{det},X}
+\overline{N}_{1_Z,1}^{\mathrm{det},X}
+\overline{N}_{0_X,1}^{\mathrm{det},X}
+\overline{N}_{1_X,1}^{\mathrm{det},X} .
$
In our analysis, we substitute the terms of the previous equation with 
\begin{equation}
\overline{N}_{j,1}^{\mathrm{det},X}:=N_{j}^{{\rm det},X}.
\end{equation}
Alternatively, tighter estimates of $\overline{N}_{j,1}^{\mathrm{det},X}$ could be obtained by applying the decoy-state method additional times, increasing the security parameter $\varepsilon_{\rm sec}$. Finally, note that rounds in which Alice prepares an $X$-basis state and Bob performs a $Z$-basis measurement are not used in the security analysis. Consequently, they do not need to be included in Azuma's deviation term.
}  
and
\begin{equation}
    \hat{N}_{Z,1}^{\mathrm{det},Z} = \frac{\overline{N}_{Z,1}^{\mathrm{det},Z}}{1-\delta_2-\gamma_{\rm bin}^{\varepsilon_{\rm conc}}(\underline{N}_{Z,1}^{\mathrm{det},Z}, \delta_2)},
\end{equation}
where the function $\gamma_{\rm bin}^{\varepsilon_{\rm conc}}(N, \delta)$ is defined in \cref{eq:gamma_bin} and $\varepsilon_{\rm conc}$ represents a failure probability.
The function $G_{+}(y,z)$ is defined as 
\begin{equation}
G_{+}(y,z)=
\begin{cases}
\begin{aligned}
&y+(1-z^2)(1-2y)\\
&\quad+2z\sqrt{(1-z^2)y(1-y)},
\end{aligned}
& \substack{0\le y<z^2,0\le z\le1} \\
1, & \text{otherwise},
\end{cases}
\end{equation}
where 
\begin{equation}
\label{eq:z}
z=1-\frac{2 p_{Z_C}\left(\overline{N}_{X_C=1,1}+\Delta_A\right)}{p_{X_C}\left[\overline{N}_{Z_C=0,1}^{\mathrm{det}}+\Delta_A+\max \left(0, \underline{N}_{Z_C=1,1}^{\mathrm{det}}-\Delta_A\right)\right]},
\end{equation}
and
\begin{equation}
\label{eq:y}
y=\frac{\overline{N}_{Z_C=1,1}^{\text {err }}+\Delta_A}{\underline{N}_{Z_C=1,1}^{\text {det }}-\Delta_A}.
\end{equation}
The quantity $\underline{N}_{0_Z, \mathrm{TAR}, Z_C,1}^{0_X}$ denotes a lower bound on the number of single-photon rounds with $0_Z$ encoding and $0_X$ measurement outcome, in which the coin is measured in $Z_C$ and the result corresponds to the tag TAR. In \cref{eq:z}, $\overline{N}_{X_C=1,1}$, represents an upper bound on the number of single-photon rounds in which $X_C=1$, $p_{X_C}$ is the probability of measuring the coin in the $X$ basis and $\underline{N}_{Z_C=1, 1}^{\rm det}$ denotes a lower bound on the number of single photon rounds with outcome $Z_C=1$, while in \cref{eq:y} $\overline{N}_{Z_C=1,1}^{\text {err }}$ represents an upper bound on the number of single-photon erroneous detections with $Z_C=1$.

In Appendix~\ref{app:list} we present the mathematical expressions for the quantities, $\overline{N}_{X_C=1,1}$, $\overline{N}_{Z_C=0,1}^{\mathrm{det}}$, $\underline{N}_{Z_C=1,1}^{\mathrm{det}}$ and $\overline{N}_{Z_C=1,1}^{\mathrm{err}}$, required to evaluate \cref{eq:number_errors}.

As a final step in the computation of the number of phase errors, one must obtain, via the decoy-state analysis, both upper and lower bounds on $N_{j,1}^{{\rm det},b_X}$, \textit{i.e.}, the number of single-photon rounds in which Alice prepares the state $j$ and Bob obtains the $X$-basis outcome $b$, with $b \in \{0,1\}$. This can be accomplished by suitably modifying the linear program in \cref{eq:LP_key}, as detailed in Appendix~\ref{app:LP}.

In the absence of detector imperfections, the desired bound on the single-photon phase-error rate in one group simply reads $\overline{e}_{{\rm ph},1}:=\overline{N}_{{\rm ph},1
}/\underline{N}_{Z,1}^{{\rm det},Z}$. In Appendix~\ref{app:detector} we show that detector imperfections can be incorporated by modifying the previous bound to
\begin{equation}
\label{eq:bound_phase_error}
\begin{aligned}
\overline{e}_{{\rm ph},1}:=\frac{\overline{N}_{{\rm ph},1
    }}{\underline{N}_{Z,1}^{{\rm det},Z}}+ \frac{\delta_1+\gamma_{\rm bin}^{\varepsilon_{\rm conc}}(\underline{N}_{Z,1}^{{\mathrm{det},Z}}, \delta_1)}{1-\delta_2-\gamma_{\rm bin}^{\varepsilon_{\rm conc}}(\underline{N}_{Z,1}^{{\mathrm{det},Z}}, \delta_2)}.
\end{aligned}
\end{equation}

\section{Simulations}\label{sec:simulations}

Before running the quantum communication phase of the protocol, Alice and Bob agree on the total number of emitted signals $N$, together with the correctness and secrecy parameters of the final key $\varepsilon_{\rm corr}$ and $\varepsilon_{\rm sec}$,
the penalty due to privacy amplification, $\varepsilon_{\rm PA}$, and the failure probability of the statistical bounds employed in the security proof. For simplicity, we set the latter to a common value $\varepsilon_{\rm conc}$. Overall, this leads to $\varepsilon_{\rm sec}=2\sqrt{\varepsilon+\zeta}+\varepsilon_{\rm PA}$, where the composition of union bounds yields
\begin{equation}
    \varepsilon= (l_c+1)(20+M)\varepsilon_{\rm conc},
\end{equation}
with $M=9[2(m_{\rm max}+1)+2\mathcal{P}+4]$ representing the security budget associated with the decoy-state estimation. Here $m_{\rm max}$ denotes the threshold photon number used in the linear programs and $\mathcal{P}$ denotes the number of intensities. Lastly, we set $\zeta=\varepsilon_{\rm conc}$ for simplicity.

In the simulations, we employ the channel model introduced in Appendix~\ref{app:channel} and the values reported in \cref{tab:values}.  For the intensity distribution described in \cref{eq:pns},
we assume that the actual intensity is uniformly distributed over the interval
$[\mu_p(1-\delta_{\mu_p}),\,\mu_p(1+\delta_{\mu_p})]$, and we set $\delta_{\mu_p}=10^{-2}$ for the relative deviation of each intensity with respect to its nominal value. The polarization misalignment is fixed to $\theta_{\rm mis}=0.1$, the state preparation error of the source to $\delta_{\rm SPF}=0.063$~\cite{Guille_framework,SPF1,SPF2} and we consider several values of $\epsilon_{\rm side}$. Regarding correlations, we adopt the model introduced in \cref{eq:corre_model} and take $\epsilon_{1}=10^{-4}$ and $C=13$, in alignment with the experimental results observed in \cite{agulleiro2025modelingcharacterizationarbitraryorder}. In addition, we evaluate the uncorrelated case $\epsilon_{1}=0$ for comparison.

With respect to the detectors, we assume a detection efficiency $\eta_{\rm det}=0.73$ \cite{Pittaluga_2021} and a dark count probability of $d_{\rm dc}=10^{-8}$ \cite{Pittaluga_2021}, and evaluate two levels of detector imperfections  $\Delta_{\rm dc}=\Delta_{\eta}=0.02$ in \cref{fig:performance}~(a) and a more pessimistic case $\Delta_{\rm dc}=\Delta_{\eta}=0.05$ in \cref{fig:performance}~(b).
Moreover, we use a block size  $N=10^{12}$ and set the secrecy and correctness parameters to $\varepsilon_{\rm sec}=\varepsilon_{\rm corr} = 10^{-10}$ and  $\varepsilon_{\rm PA}=\varepsilon_{\rm sec}/3$. Lastly, we set the number of bits used for error correction to $\lambda_{\rm EC}=1.16 \times {N}_{Z}^{{\rm det},Z} h(e_{Z})$, where $e_{Z}$ denotes the overall quantum bit error rate in the $Z$ basis arising from the channel model.

\begin{table}[]
\centering
\small
\setlength{\tabcolsep}{4pt}
\renewcommand{\arraystretch}{1.1}
\begin{tabular}{cc}
\hline
\textbf{Parameter} & \textbf{Value} \\
\hline
$\theta_{\rm mis}$ & $0.1$ \\
$\delta_{\mu_p }$& $10^{-2}$\\
$\eta_{\rm det}$ & $0.73$ \\
$d_{\rm dc}$ & $10^{-8}$ \\
$\varepsilon_{\rm sec}=\varepsilon_{\rm corr}$ & $10^{-10}$ \\
$\varepsilon_{\rm PA}$ & $10^{-10}/3$ \\
$\mu_3$ & $5\times 10^{-4}$ \\
$p_{\mu_3}$ & $0.05$ \\
\hline
\end{tabular}
\caption{Fixed parameters used for the simulations in this section and those in Appendix~\ref{app:numerical_sim}. $\theta_{\rm mis}$ denotes the polarization misalignment, $\delta_{\mu_p }$ the relative deviation of each intensity with respect to its nominal value, $\eta_{\rm det}$ the detector efficiency and $d_{\rm dc}$ is the dark count probability. Furthermore, $\varepsilon_{\rm sec}$ and $\varepsilon_{\rm corr}$ represent the secrecy and correctness parameters respectively, and $\varepsilon_{\rm PA}$ the privacy amplification parameter. Finally, $\mu_3$ denotes the weaker intensity in the decoy set which is emitted with probability $p_{\mu_3}$.  }
\label{tab:values}
\end{table}

Finally, in the figures, for each value of the distance, we optimize the basis probabilities $p_{Z_A}=p_{Z_B}$, the probability $p_{Z_C}$, as well as the intensity settings and their probabilities, fixing the lowest out of the three settings to $\mu_3=5 \times 10^{-4}$ and $p_{\mu_3}=0.05$ for simplicity. 

\begin{figure}[]
    \centering
    \includegraphics[width=\columnwidth]{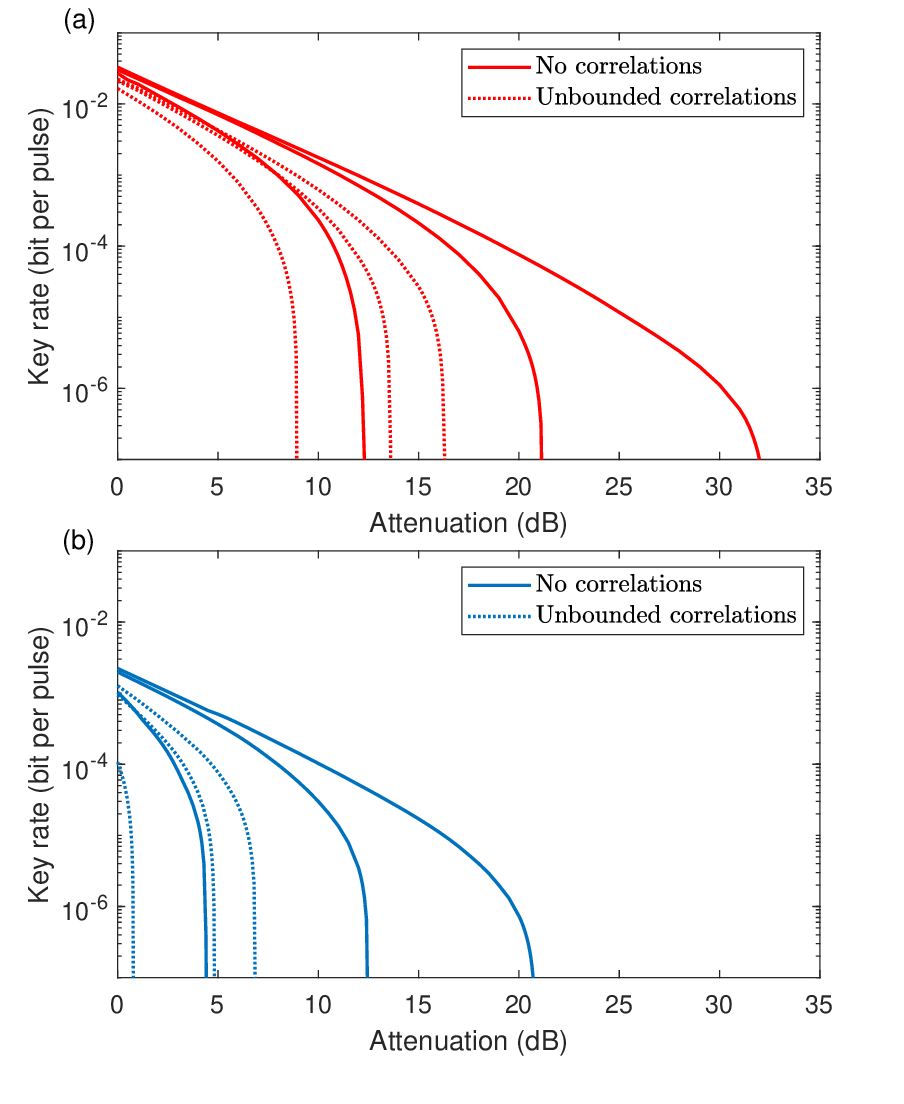}
    \caption{Finite secret-key rate of a decoy-state BB84 protocol incorporating both source and detector imperfections as a function of the channel loss in dB. In both figures, we fix the state-preparation error parameter to $\delta_{\rm SPF}=0.063$~\cite{Guille_framework,SPF1,SPF2} and consider three values for the magnitude of the side channel leakage, $\epsilon_{\rm side}\in\{10^{-3},10^{-4},10^{-6}\}$ corresponding to the three lines of each type shown in each figure. The two figures differ in the magnitude of the detector imperfections. Figure~(a) corresponds to $\Delta_{\eta}=\Delta_{\rm dc}=0.02$, while figure~(b) illustrates the more pessimistic scenario $\Delta_{\eta}=\Delta_{\rm dc}=0.05$.}\label{fig:performance}
    
\end{figure}

The results are illustrated in \cref{fig:performance}. We find that secret key generation remains possible for side-channel leakage values up to about $\epsilon_{\rm side}=10^{-3}$, even in the presence of correlations and detector imperfections. Our numerical results also reveal a strong dependence of the secret key rate on the parameters $\Delta_{\eta}=\Delta_{\rm dc}$, which characterize the magnitude of the detector imperfections. In Appendix~\ref{app:numerical_sim} we present further simulations in which we evaluate the secret key rate for different values of side channel leakage $\epsilon_{\rm side}$, state preparation flaws $\delta_{\rm SPF}$, correlation strengths $\epsilon_1$, detector tolerances $\Delta_{\rm dc}$ and $\Delta_{\eta}$ and block sizes $N$. As expected, the main parameters limiting performance are side-channel leakage and bit and basis correlations because they effectively increase the dimensionality of the emitted quantum states, making the system vulnerable against unambiguous state discrimination attacks \cite{USD,USD2}. As a result, the magnitude of these imperfections must be really tiny to achieve long transmission distances. Also, as already mentioned, the magnitude of detector tolerances has also a significant impact. On the other hand, state preparation flaws have a much lower impact on the secret key rate, mainly caused by the increase in the quantum bit error rate with the magnitude of the parameter $\delta_{\rm SPF}$ according to our channel model. In \cref{app:numerical_sim}, we also show that positive secret-key rates can still be achieved even for smaller block sizes of $N=10^{10}$ in the presence of all the imperfections considered above given that their magnitude is sufficiently small. 

\section{Conclusions}\label{sec:conclusions}

In this work, we have investigated the security of decoy-state QKD protocols in the finite-key regime in the presence of multiple device imperfections, including state-preparation flaws, bit and basis correlations and information leakage, setting-independent intensity fluctuations and detector imperfections. By incorporating all these effects together within a unified analytical security framework, our analysis provides a comprehensive assessment of practical implementations, where deviations from idealized device model assumptions are unavoidable and typically occur simultaneously. Our results show that secret key generation remains possible even under these non-ideal conditions, although at the cost of a significant performance degradation. This penalty is mainly due to the fact that some of these imperfections increase the dimensionality of the emitted quantum signals making the system vulnerable against unambiguous state discrimination attacks. In addition, the intrinsic conservativeness of the analytical security proof considered, which accommodates several worse-case assumptions at once, might make the parameter estimation procedure rather loose. In this regard, the development of numerical finite-key techniques capable of enabling tighter parameter estimation and more optimized statistical bounds constitutes a promising direction to mitigate the observed secret-key rate penalty and enhance the achievable performance in realistic scenarios \cite{Navarretenumerical}. Furthermore, incorporating new methods for handling correlations that avoid the statistical penalty associated with partitioning the protocol into groups, such as those recently proposed in \cite{nogroups}, could further improve performance.

In addition, our numerical simulations suggest that the achievable secret key rate is strongly affected by detector imperfections. Therefore, improving their analytical treatment or using interference-based QKD protocols \cite{MDI,twin}, which remove the need to trust the detection process, emerge as a promising alternative strategy to improve performance. 

An important direction for future work is the inclusion of decoy-state related imperfections affecting the intensity modulation and the global phase-randomization process, beyond setting-independent intensity fluctuations. These imperfections typically break the counterfactual scenario underlying standard security proofs, rendering our proof inapplicable in its current form. 

Overall, our results contribute to bridging the gap between theoretical security analyses and realistic QKD implementations, and provide further evidence that secure quantum communication remains achievable under experimentally relevant conditions.

\section*{Acknowledgments}

We thank Vaisakh Mannalath for helpful discussions and assistance with the code used for the simulations. This work was supported by the Galician Regional Government (consolidation of Research Units: AtlantTIC); the Spanish Ministry of Science, Innovation and Universities (MICIU); the Fondo Europeo de Desarrollo Regional (FEDER) through the grant No. PID2024-162270OB-I00; the "Hub Nacional de Excelencia en Comunicaciones Cuánticas" funded by the Ministerio para la Transformación Digital y de la Función Pública and the European Union NextGenerationEU; the European Union's Horizon Europe Framework Programme under the Marie Sklodowska-Curie Grant No. 101072637 (Project QSI); the project "Quantum Secure Networks Partnership" (QSNP, grant agreement No 101114043); the European Union under the Project IberianQCI (grant 101249593), and the Programa de Cooperación Interreg VI-A España–Portugal (POCTEP) 2021–2027 through the project QUANTUM IBERIA. A.N. acknowledges financial support from the Xunta de Galicia (Consellería de Educación, Ciencia, Universidades e Formación Profesional) through a Xunta de Galicia Postdoctoral Fellowship (No. ED481B-2025/113). G.C.-L. acknowledges funding from the European Union's Horizon Europe research and innovation programme under the Marie Sklodowska-Curie. X.~S. acknowledges support from an FPI scholarship (PRE2021-097170) granted by the Spanish Ministry of Science and Innovation.

\section*{Author contributions}

M.C. identified the need for the research project. X.S., G.C.-L. and M.P. developed the main security proof, with input from all authors. V.Z., Á.N., M.P. and G.C.-L. developed the code used to perform the numerical simulations. X.S. wrote the manuscript with contributions from M.C., and all authors contributed to its revision and to checking the validity of the results.

\section*{Code availability statement}

The code used for the numerical simulations presented in this study is not publicly available, as it is currently undergoing software registration for intellectual property protection. The code is available upon reasonable request for peer review and academic verification.

\section*{Data availability statement}

No new data was created or analysed in this study.

\section*{Competing interest}

The authors declare no competing interest.

\newpage
\appendix

\begin{widetext}

\section{Extension of \cite{curraslorenzo2026rigorousphaseerrorestimationsecurityframework} to decoy-state protocols } \label{app:rigurous}

Here we show how to extend the techniques developed in \cite{curraslorenzo2026rigorousphaseerrorestimationsecurityframework} to estimate the phase-error rate of decoy-state protocols with bit and basis correlations. This approach allows to determine the phase-error rate of the overall protocol by combining the phase-error rates of the corresponding uncorrelated subprotocols. Precisely, we extend Theorem 2, Corollary 1, and Corollary 2 of \cite{curraslorenzo2026rigorousphaseerrorestimationsecurityframework} to the decoy-state setting, which correspond respectively to \cref{thm:2}, \cref{coro:1} and \cref{coro:2} below.

\subsection{Protocol description and security framework}

We start by reviewing the different steps of the actual protocol and its equivalent source-replaced description. For generality, the following description is given for generic prepare-and-measure protocols rather than for the specific BB84 implementation of the main text, so that the results also apply for instance to  three-state protocols \cite{boileauUnconditionalSecurity2005,loss_tolerant}.

In each round $k\in\{1,2,...,N\}$ of the protocol Alice selects a certain encoding setting $j_k\in\mathcal{J}$ with probability $p_{j_k}$ and an intensity setting $\mu_k$  with probability $p_{\mu_k}$. In the presence of encoding correlations affecting the bit and basis information, the state emitted in round $k$ depends not only on $j_k$ but also on all previous settings, $j^{k-1}_1=j_{k-1}...j_1$. We denote Alice's basis selection in round $k$ as $\tilde{B}_k$, and for simplicity we assume $\tilde{B}_k=\mathrm{key}$ when $j_k \in \{0,1\}$ .

On the receiving side, Bob chooses $\xi_k\in\{{\rm key, test}\}$ with probability $p_{\xi_k}$ and performs the POVM $\vec{\Gamma}_{\xi_k}$. The POVM $\vec{\Gamma}_{\rm key}:=\left\{\Gamma_0^{\rm key}, \Gamma_1^{\rm key}, \Gamma_{\perp}^{\rm key}\right\}$ has three elements, corresponding respectively to bit 0 and 1 and a non-detection denoted by $\perp$,
while $\vec{\Gamma}_{\rm test}$ can be a POVM with any number of outcomes. We consider that Alice and Bob extract their sifted keys from the rounds in which $\tilde{B}_k=\xi_k=\rm key$ and Bob obtains a detection, that is, his outcome is 0 or 1 rather than $\perp$. With this, we define the following protocol:

\begin{tcolorbox}[breakable, colback=brown!1,
  colframe=brown!30,title=Actual protocol]
\begin{enumerate}
\item \textbf{State preparation:} For each round $k \in\{1, \ldots, N\}$, Alice randomly selects a certain encoding setting $j_k$ with probability $p_{j_k}$ and an intensity setting $\mu_k$ with probability $p_{\mu_k}$, and then prepares a quantum state $\rho_{\mu_k,j_k|j_{1}^{k-1}}$. For the BB84 protocol, this state corresponds to \cref{eq:global_state} in the main text. She sends this state to Bob through the quantum channel.

\item \textbf{Eve's action:} Eve applies an arbitrary quantum operation on all the transmitted states and re-sends some output systems $B_1^N$ to Bob, while keeping some ancillary system $E$.
\item \textbf{Bob's measurements:} For each round $k$, Bob chooses $\xi_k \in\{{\rm key, test}\}$ with probability $p_{\xi_k}$ and performs the corresponding POVM; $\vec{\Gamma}_{\text {key }}$ if $\xi_k= {\rm key}$, or $\vec{\Gamma}_{\text {test }}$ if $\xi_k={\rm test}$. Bob records his measurement outcomes.

\item \textbf{Sifting:} Bob announces which rounds resulted in a detection, along with his choice of $\xi_k$ for each detected round.
For the detected rounds, Alice announces the intensity setting $\mu_k$. Also, for the detected rounds with $\xi_k=$ test, and for the detected rounds with $\xi_k=$ key and {$j_k \notin\{0,1\}$}, Alice announces $j_k$. For the detected rounds with $\xi_k=$ key and {$j_k \in\{0,1\}$}, Alice announces $\tilde{B}_k=$ key.

\item \textbf{Bit-error-rate estimation:} Alice and Bob choose a random subset, denoted by $\mathcal{D}_{\text {key}}$, of the detected rounds in which $\tilde{B}_k=\xi_k=\mathrm{key}$, which will be used to generate the sifted key \footnote{Note that in the protocol considered in the main text, Alice and Bob extract key from all detected rounds in which they select the $Z$ basis. For the sake of generality, in this protocol description we consider the possibility that only a random subset of the detected rounds in which $\tilde{B}_k=\xi_k=\mathrm{key}$ is used for sifted key generation.}. They announce which rounds belong to this subset. For the remaining rounds with $\tilde{B}_k=\xi_k={\rm key}$, Alice and Bob announce their bit values to estimate the bit-error rate.

\item \textbf{Sifted key formation:} For each round $k \in \mathcal{D}_{\text {key }}$, Alice's sifted key bit is $j_k$ and Bob's sifted key bit is his measurement outcome, which we denote by $b_k$.

\item \textbf{Variable-length decision:} Let $\vec{n}$ denote all the data announced by Alice and Bob in the previous steps. Using this data, Alice and Bob compute the number of bits $\lambda_{\mathrm{EC}}$ to be revealed for one-way error correction and the length of the final key $\ell$, which includes the estimation of all relevant parameters required to determine $\ell$. Aborting corresponds to $\ell=0$.

\item \textbf{Error correction and verification:} They perform one-way error correction revealing $\lambda_{\mathrm{EC}}$ bits, and verify correctness using a randomly chosen hash function from a two-universal family with output length $\log_{2}(2/\varepsilon_{\mathrm{corr}})$, with one party announcing the hash value for comparison.

\item \textbf{Privacy amplification:} Given that error verification succeeds, they apply a randomly selected two-universal hash function to the sifted key to obtain a final key of length $\ell$.

\end{enumerate}
\end{tcolorbox}

The security of the above protocol is analyzed using the source-replacement technique. In this picture, Alice’s prepare-and-measure procedure is equivalent to one in which she first prepares the entangled state 
\begin{equation}
\label{eq:state_app}
\left|\Psi_N\right\rangle_{A_1^N I_1^N M_1^N T_1^N}=\sum_{j_1^N} \sqrt{\operatorname{Pr}\left[j_1^N\right]}\left|j_1^N\right\rangle_{A_1^N}\bigotimes_{k=1}^N\sum_{\mu_k} \sqrt{p_{\mu_k}}|{\mu}_{k}\rangle_{I_k} \sum_{m_{k}=0}^{\infty} \sqrt{p_{m_k \mid \mu_k}}|m_k\rangle_{M_k}\left|\psi_{j_1^k}^{m_k,k}\right\rangle_{T_k},
\end{equation}
then she sends the photonic systems $T^N_1= T_1...T_N$ through the quantum channel to Bob, and measures the ancillary systems $A_k$, $I_k$ and $M_k$ to determine her setting choices and the photon number for each emitted signal.

Bob's key POVM can be decomposed in two separate steps. First, he applies a filter operation $\{F, \mathbb{I}-F\}$ with $F=\Gamma_0^{\text {key }}+\Gamma_1^{\text {key }}$ that determines whether a detection occurs. After that, he applies  a two-outcome POVM $\left\{G_0, G_1\right\}$ acting on the support of $F$ that determines the bit value conditional on detection (see \cite{curraslorenzo2026rigorousphaseerrorestimationsecurityframework} for more details). This decomposition preserves the measurement statistics and allows Bob’s bit-value measurement to be deferred, facilitating the comparison with the phase-error estimation protocol presented below. Following this argument, we define a source-replaced protocol that is statistically equivalent to the actual prepare-and-measure protocol described above by construction.

\begin{tcolorbox}[breakable, colback=blue!1,
  colframe=teal!30,title=Source-replaced protocol]

\begin{enumerate}
\item \textbf{State preparation:} Alice prepares the global entangled state $\left|\Psi_N\right\rangle_{A_1^N I_1^N M_1^N T_1^N}$  given by \cref{eq:state_app} and sends the photonic systems $T_1^N$ through the quantum channel to Bob, while retaining the registers $A_1^N, I_1^N$ and  $M_1^N$. 

\item \textbf{Eve's action:} Eve applies an arbitrary quantum operation on all the transmitted states $T_{1}^{N}$, and sends the output systems $B_{1}^{N}$ to Bob, possibly retaining an ancillary system $E$.

\item \textbf{Detection and test:} For each round, Bob selects $\xi_k \in \{\text{key}, \text{test}\}$ with probability $p_{\xi_k}$. If $\xi_k = \text{key}$, he applies the filter operation $\{F, \mathbb{I} - F\}$ to determine whether a detection occurs. If $\xi_k = \text{test}$, he performs the measurement $\vec{\Gamma}_{\text{test}}$ (a POVM with an arbitrary number of outcomes prescribed by the actual protocol). Bob announces which rounds result in a detection, together with his choice of $\xi_k$ for the detected rounds.  

\item \textbf{Key/test determining, setting announcement and photon number measurement:} For the detected rounds, Alice proceeds as follows. She measures her system $I_k$ in the basis $\ket{{\mu}_k}_{I_k}$ and announces her results. Also, if $\xi_k=$ test,
Alice measures her system $A_k$ in the $\left\{\left|j_k\right\rangle_{A_k}\right\}_{j_k \in \mathcal{J}}$ basis and announces her outcome $j_k$. If $\xi_k=$ key, she performs a projection onto the subspace spanned by $\left\{|0\rangle_{A_k},|1\rangle_{A_k}\right\}$. If successful, Alice announces $\tilde{B}_k=$ key. If not, Alice measures her system $A_k$ in the $\left\{\left|j_k\right\rangle_{A_k}\right\}_{j_k \in \mathcal{J}}$ basis and announces her outcome $j_k$. Furthermore, she measures system $M_k$, thereby obtaining the value of $m_k$ in every round. Note that this latter measurement does not alter either the statistics or the announcements with respect to the real protocol.
\item \textbf{Bit-error-rate estimation:} Alice and Bob choose a random subset $\mathcal{D}_{\text {key }}$ from the set of  detected rounds in which $\tilde{B}_k=\xi_k=$ key, which will be used to generate the sifted key, and announce this information. For the remaining rounds in which $\tilde{B}_k=\xi_k=$ key, Alice measures $A_k$ in $\left\{|0\rangle_{A_k},|1\rangle_{A_k}\right\}$ and Bob measures his filtered system using $\left\{G_0, G_1\right\}$, a two-outcome POVM that determines the bit value conditional on detection. Both announce their results.
\item \textbf{Sifted-key measurements:} For each round $k \in \mathcal{D}_{\text {key }}$, Alice measures $A_k$ in $\left\{|0\rangle_{A_k},|1\rangle_{A_k}\right\}$ and Bob measures his filtered system using $\left\{G_0, G_1\right\}$.  The sifted key is defined by their respective bit outcomes in these rounds.

\item[7--9.] Same as in the Actual protocol.

\end{enumerate}
\end{tcolorbox}

Finally, before stating the main result, we also define the Phase-error estimation protocol, which is used to prove the security of the final key pair by computing the single-photon phase-error rate.

\begin{tcolorbox}[breakable, colback=magenta!1,
  colframe=magenta!30,title=Phase-error estimation protocol]
\begin{enumerate}
\item [1--5.] Same as in the Source-replaced protocol.
\item [6.] {\bf Single-photon phase-error measurement:} For each round $k \in \mathcal{D}_{\text {key }}$, Alice measures $A_k$ in $\left\{|+\rangle_{A_k},|-\rangle_{A_k}\right\}$, where $| \pm\rangle_{A_k}= \left(|0\rangle_{A_k} \pm|1\rangle_{A_k}\right) / \sqrt{2}$, and Bob measures his filtered system using $\left\{G_{+}, G_{-}\right\}$. We denote the single-photon phase-error rate $e_{{\rm ph,1}}$ as the fraction of events in which their outcomes differ among the rounds in which $m_k=1$.
\end{enumerate}
\end{tcolorbox}

With this latter protocol we can find bounds of the form
\begin{equation}
\label{eq:phase_EUR}
\operatorname{Pr}\left[N_{{\rm key},1}^{\rm det}<\underline{N}_{{\rm key},1}^{\rm det} \cup e_{{\rm ph},1}>g_{\epsilon}(\overrightarrow{n}, N)\right] \leq \epsilon,
\end{equation}
where $N_{{\rm key},1}^{\rm det}$ represents the number of single-photon detected key rounds, ${e}_{{\rm p h},1}$ represents the single-photon phase-error rate, $\overrightarrow{{n}}$ is the  vector containing the announced data before step 6, $N$ is the total number of transmitted rounds, $\epsilon'$ is the failure probability of the bound, and $g_{\epsilon}(\cdot,\cdot)$ is a function that relates all these quantities. Bounds of the form given in \cref{eq:phase_EUR} are sufficient to establish security through either phase-error correction (PEC) techniques \cite{Koashi} or by combining entropic uncertainty relations (EUR) together with the leftover hashing lemma (LHL) \cite{EUR_1,EUR_2,EUR_3}, even in the case of variable-length keys \cite{Tupkary2025phaseerrorrate}. In particular, from the EUR+LHL framework, we have the following result.
\begin{theorem}[Variable-length security of decoy-state-type protocols via EUR]
\label{thm:variable_length_security_decoy}

Suppose that we have the guarantee that, in the phase-error estimation protocol described above, \cref{eq:phase_EUR} holds, where $\underline{N}_{{\rm key},1}^{\rm det}$ and $g_{\epsilon}(\overrightarrow{n}, N)$ are functions of the observed data vector $\overrightarrow n$. Let $\lambda_\mathrm{EC}$ be a function of $\overrightarrow n$ that determines the number of bits disclosed during error correction, and let
\begin{equation}
\label{eq:final_key_length}
    \ell = \max\Big(0,\underline{N}_{{\rm key},1}^{\rm det} \big[1-h\big(g_{\epsilon}(\overrightarrow{n}, N)\big)\big] - \lambda_\mathrm{EC} - 2 \log_2 \frac{1}{2\varepsilon_\mathrm{PA}}-\log_2 \frac{2}{\varepsilon_\mathrm{corr}}\Big),
\end{equation}
be a function of $\overrightarrow n$ that determines the length of the final key, where $h(x)$ is the binary entropy function for $x\leq 1/2$ and $h(x) = 1$ otherwise. Then, if Alice and Bob run the \emph{Actual protocol (source replaced)} using this choice of $\lambda_\mathrm{EC}$ and $\ell$, the output key is $(\varepsilon_{\mathrm{corr}}+\varepsilon_{\mathrm{sec}})$-secure, with $\varepsilon_{\mathrm{sec}} = 2\sqrt{\epsilon} + \varepsilon_{\mathrm{PA}}$. 
\end{theorem}

\begin{proof}
    Essentially identical to \cite[Theorem 1]{curras-lorenzoSecurityDecoystate2026}, see also \cite[Theorem 3]{Tupkary2025phaseerrorrate}.
\end{proof}

\subsection{Technical results}

Having established that suitable bounds on the number of single-photon key rounds and the single-photon phase-error rate suffice to guarantee security, we now show how to bound the latter in the presence of bit and basis correlations of range $l_c$. For that, we partition the protocol rounds into $(l_c+1)$ groups, where rounds within each group are separated by at least $l_c$ positions and thus uncorrelated. This allows us to compute bounds on the single-photon phase-error rate of each group independently. Crucially, these per-group bounds can be combined directly into a single overall phase-error bound, avoiding the need to perform separate privacy amplification steps nor to invoke composability arguments. As highlighted in the beginning of the section, the results below generalize respectively Theorem 2, Corollary 1, and Corollary 2 of \cite{curraslorenzo2026rigorousphaseerrorestimationsecurityframework} to the decoy-state setting.

\begin{theorem} 
\label{thm:2}
Consider a prepare-and-measure decoy-state QKD protocol with an uncorrelated source in which Alice emits a global state
\begin{equation}
\label{eq:uncorrelated_state_general}
\left|\Psi_N\right\rangle_{A_1^N I_1^N M_1^N T_1^N}=\sum_{j_1^N} \sqrt{\operatorname{Pr}\left[j_1^N\right]}\left|j_1^N\right\rangle_{A_1^N}\left|\Psi_{j_1^N}\right\rangle_{I_1^N M_1^N T_1^N},
\end{equation}
where $j_1^N=j_1...j_N$, $\operatorname{Pr}\left[j_1^N\right]=\prod_{k=1}^{N}p_{j_k}$,  and
\begin{equation}
\label{eq:uncorrelated_state}
\left|\Psi_{j_1^N}\right\rangle_{I_1^N M_1^N T_1^N}=\bigotimes_{k=1}^N\sum_{\mu_k} \sqrt{p_{\mu_k}}|{\mu}_{k}\rangle_{I_k} \sum_{m_{k}=0}^{\infty} \sqrt{p_{m_k \mid \mu_k}}|m_k\rangle_{M_k}\left|\psi_{j_k}^{m_k,k}\right\rangle_{T_k}.
\end{equation}
Suppose that there exists an admissibility set $S^{N}$ for each $N$, such that as long as $\Big\{\left|\Psi_{j_1^N}\right\rangle_{I_1^N M_1^N T_1^N}\Big\}_{j_{1}^{N}\in \mathcal{J}^{N}}\in S^{N}$ for any sequence of setting choices $j_{1}^{N}$ and for any attack performed by Eve, the following condition holds,  
\begin{equation}
\label{eq:phase_error_basic}
\operatorname{Pr}\left[N_{{\rm key},1}^{{\rm det}}<\underline{N}_{{\rm key},1}^{{\rm det}} \cup e_{{\rm ph},1}>g_{\epsilon}(\overrightarrow{n}, N)\right] \leq \epsilon,
\end{equation}
for a certain $\underline{N}_{{\rm key},1}^{{\rm det}}$ and  $g_{\epsilon}(\cdot,\cdot)$ and where $\overrightarrow{n}$ is a vector containing the announced data (before post-processing).

Besides,  consider the analogous protocol with a source that suffers from correlations in bit and basis settings. That is, in this scenario Alice emits the state 
\begin{equation}
\begin{aligned}
\label{eq:state_correlated}
&\left|\Psi_N^{\prime}\right\rangle_{A_1^N I_1^N M_1^N T_1^N}=\sum_{j_1^N} \sqrt{\operatorname{Pr}\left[j_1^N\right]}\left|j_1^N\right\rangle_{A_1^N}\left|\Psi_{j_1^N}^{\prime}\right\rangle_{I_1^N M_1^N T_1^N},
\end{aligned}
\end{equation}
where 
\begin{equation}
\begin{aligned}
&\left|\Psi_{j_1^N}^{\prime}\right\rangle_{I_1^N M_1^N T_1^N}=\bigotimes_{k=1}^N\sum_{\mu_k}\sqrt{p_{\mu_k}}|{\mu}_{k}\rangle_{I_k} \sum_{m_{k}=0}^{\infty} \sqrt{p_{m_k \mid \mu_k}}|m_k\rangle_{M_k}\ket{\psi_{ j_{k-l_c}^{k}}^{m_k,k}}_{T_k}.
\end{aligned}
\end{equation}
In this latter protocol with correlated states let us partition the rounds into $(l_c+1)$ sets or subprotocols according to $P_{w}=\{k:k\equiv w \text{ mod } l_c+1\}$ and define the complementary sets $P_{\bar{w}}=\{k:k\not\equiv w \text{ mod } l_c+1\}$. Then, for each $w$ define the subsequences $j_{P_{w}}$ and $j_{P_{\bar{w}}}$ indexed by $P_{w}$ and $P_{\bar{w}}$ respectively. Suppose that for every $w$ and $j_{P_{\bar{w}}}$ there exists an isometry $V^{j_{P_{\bar{w}}}}_{T_{P_w} \rightarrow I_{P_{\bar{w}}} M_{P_{\bar{w}}} T_1^N} : \mathcal{H}_{T_{P_w}}\rightarrow \mathcal{H}_{I_{P_{\bar{w}}} M_{P_{\bar{w}}}T_{1}^{N}}$ such that
\begin{equation}
\label{eq:isomety_existance}
    \Big\{(V^{j_{P_{\bar{w}}}}_{T_{P_w} \rightarrow I_{P_{\bar{w}}} M_{P_{\bar{w}}} T_1^N})^{\dagger}  \left|\Psi_{j_1^N}^{\prime}\right\rangle_{I_1^N M_1^N T_1^N}\Big\}_{j_{P_w}\in \mathcal{J}^{N,w}}\in S^{N,w},
\end{equation}
where $S^{N,w}$ represents the admissibility set of size $N^{w}\equiv |P_{w}|$.

Then,
\begin{equation}
\operatorname{Pr}\left[N_{{\rm key},1}^{{\rm det}}<\sum_{w}\underline{N}_{{\rm key},1}^{{\rm det},w} \cup e_{{\rm ph},1}>\max_{w} g_{\epsilon}(\overrightarrow{n}^{w},N^{w})\right]\leq(l_c+1)\epsilon.
\end{equation}
\end{theorem}

\begin{proof}

Let us first consider the protocol in the absence of correlations, in which Alice generates the state given by \cref{eq:uncorrelated_state}.  Eve applies a global isometry $V_{T_{1}^N\rightarrow B_{1}^{N}E}$ and sends systems $B_{1}^{N}$ to Bob, while keeping system $E$. We are interested in the state shared by Alice and Bob after Eve’s attack, that is, after step two of the source-replaced protocol, and thus we define a completely-positive trace-preserving (CPTP) map $\Phi_{T_{1}^N\rightarrow B_{1}^{N}}$ consisting of first applying the global isometry and then tracing out system $E$. The state shared by Alice and Bob after Eve’s attack reads
\begin{equation}
    \rho_{A_1^N I_1^N M_1^N B_1^N}= \Phi_{T_{1}^N\rightarrow B_{1}^{N}}\Big( \ket{\Psi_N}\bra{\Psi_N}_{A_1^N I_1^N M_1^N T_1^N} \Big).
\end{equation}

Next, Alice and Bob perform measurements on their systems through which they learn $\overrightarrow{n}$, $e_{\text{ph},1}$ and $N^{\text{det}}_{\text{key},1}$ . 

We can define a simple two-outcome POVM $\left\{M_{A_1^N I_1^N M_1^N B_1^N}^{1}, M_{A_1^N I_1^N M_1^N  B_1^N}^{2}\right\}$ that only checks whether $N_{{\rm key},1}^{{\rm det}}<\underline{N}_{{\rm key},1}^{{\rm det}}\cup e_{\text{ph},1} > g_{\epsilon}(\overrightarrow{{n}},N)$ or not. 
Using this, we can restate the bound in \cref{eq:phase_error_basic} as the following guarantee: if Alice generates the global state given by \cref{eq:uncorrelated_state} and such state is in $S^{N}$, then, for any CPTP map $\Phi_{T_1^N \rightarrow B_1^N}$,
\begin{equation}
\label{eq:target}
\operatorname{Tr}\left[M_{A_1^N I_1^N M_1^N B_1^N}^{1} \Phi_{T_1^N \rightarrow B_1^N}\left(|\Psi_N\rangle\langle\Psi_N|_{A_1^N I_1^N M_1^N T_1^N}\right)\right] \leq \epsilon .
\end{equation}

Now, let's consider an analogous protocol with correlated bit and basis sources, whose global state is given by \cref{eq:state_correlated}.
The strategy to lower bound the number of single-photon key rounds and upper bound the phase-error rate in the correlated scenario is, as already mentioned, to partition the protocol rounds into ($l_c+1$) subprotocols $P_w$ indexed by $w \in\left\{0,1, \ldots, l_c\right\}$, and finding suitable bounds in each subprotocol $P_w$ independently. To achieve this, we show below that each subprotocol satisfies the conditions of the uncorrelated scenario when conditioning on any value of the settings $j_{P_{\bar{w}}}$ of the complementary set $P_{\bar{w}}$.  As a tool, we introduce a modified protocol in which Alice and Bob perform the phase-error measurements for the sifted key rounds in $P_{w}$, but perform the actual bit measurements for the rounds in $P_{\bar{w}}$. Below we show the changes between the actual source-replaced protocol (see \cite{curraslorenzo2026rigorousphaseerrorestimationsecurityframework} for more details) and the new phase-estimation protocol for the $w$-subprotocol.

\begin{tcolorbox}[breakable, colback=red!1,
  colframe=gray!30,title=$w$-th phase-error estimation protocol (PEEP) (defined for each $w$)]
\begin{enumerate}
\item[1--5.] Same as in the Source-replaced protocol.
\item[6.] \textbf{Measurements in sifted-key rounds:} Define the set of rounds $P_w=\left\{k: k \equiv w\left(\bmod l_c+1\right)\right\}$ and its complement $P_{\bar{w}}=\left\{k: k \not \equiv w\left(\bmod l_c+1\right)\right\}$.

(a) Bit measurements for rounds in $P_{\bar{w}}$ : For each round $k \in \mathcal{D}_{\text {key }} \cap P_{\bar{w}}$, Alice measures $A_k$ in $\left\{|0\rangle_{A_k},|1\rangle_{A_k}\right\}$ and Bob measures his filtered system using $\left\{G_0, G_1\right\}$.

(b) Phase-error measurements for rounds in $P_w$ : For each round $k \in \mathcal{D}_{\text {key }} \cap P_w$, Alice measures $A_k$ in $\left\{|+\rangle_{A_k},|-\rangle_{A_k}\right\}$ and Bob measures his filtered system using $\left\{G_{+}, G_{-}\right\}$. We denote the single-photon phase-error rate of the $w$-th partition $e_{\text{ph},1}^{w}$ as the fraction of single-photon events in which their outcomes differ and $m_k=1$. Note that this is well-defined, as in step 4 Alice has already measured the register $M_k$ for all rounds.
\end{enumerate}
\end{tcolorbox}

To lower bound the number of single-photon key rounds $N_{{\rm key},1}^{{\rm det}, w}$ and upper bound the single-photon phase-error rate $e_{\text{ph,1}}^{w}$ of each subprotocol $P_w$, we consider a similar protocol to the one introduced above but in which  Alice and Bob perform
their actions in a different order. First, Alice generates the global state in \cref{eq:state_correlated} and then Eve applies a certain global isometry, which can be seen as a CPTP map $\Phi^{'}_{T_1^N\rightarrow B_1^N}$ after tracing out Eve's register. Next, Alice performs all her bit and basis measurements for the rounds in $P_{\bar{w}}$ learning her setting choices $j_{P_{\bar{w}}}$. Then Alice
and Bob perform their measurements for the rounds in $P_w$ learning $\overrightarrow{n}^{w}$,  $e_{\text{ph,1}}^{w}$ and $N_{{\rm key},1}^{{\rm det}, w}$. Finally, Bob performs all his measurements for the rounds in $P_{\bar{w}}$, while Alice measures registers $I_k$ and $M_k$ for these rounds. For this modified protocol the unnormalized state shared by Alice and Bob after Alice measures the bit and basis values in $P_{\bar{w}}$ conditional on obtaining a setting choice sequence $j_{P_{\bar{w}}}$ after tracing out all the systems for the $P_{\bar{w}}$ rounds reads
\begin{equation}
\label{eq:unn_state}
\begin{aligned}
    \tilde{\rho}^{', j_{P_{\bar{w}}}}_{A_{P_w} I_{P_w} M_{P_w} B_{P_w}}  & = \operatorname{Tr}_{I_{P_{\bar{w}}} M_{P_{\bar{w}}} B_{P_{\bar{w}}}}  \Big[\bra{j_{P_{\bar{w}}}}_{A_{P_{\bar{w}}}} \Phi^{'}_{T_1^N\rightarrow B_1^N}\Big( \ket{\Psi^{'}_N}\bra{\Psi^{'}_N}_{A_1^N I_1^N M_1^N T_1^N}\Big)\ket{j_{P_{\bar{w}}}}_{A_{P_{\bar{w}}}}\Big] \\ 
    & = \text{Pr}[j_{P_{\bar{w}}}] \operatorname{Tr}_{I_{P_{\bar{w}}} M_{P_{\bar{w}}} B_{P_{\bar{w}}}}  \Big[ \Phi^{'}_{T_1^N\rightarrow B_1^N}\Big( \ket{\Psi^{''}_{j_{P_{\bar{w}}}}}\bra{\Psi^{''}_{j_{P_{\bar{w}}}}}_{A_{P_w} I_1^N M_1^N T_1^N}\Big)\Big]\\
    & = \text{Pr}[j_{P_{\bar{w}}}] \Phi^{''}_{I_{P_{\bar{w}}} M_{P_{\bar{w}}} T_1^N\rightarrow B_{P_w}}\Big( \ket{\Psi^{''}_{j_{P_{\bar{w}}}}}\bra{\Psi^{''}_{j_{P_{\bar{w}}}}}_{A_{P_w} I_1^N M_1^N T_1^N}\Big),
\end{aligned}
\end{equation}
where
\begin{equation}
    \Phi^{''}_{I_{P_{\bar{w}}} M_{P_{\bar{w}}} T_1^N\rightarrow B_{P_w}} = \operatorname{Tr}_{I_{P_{\bar{w}}} M_{P_{\bar{w}}} B_{P_{\bar{w}}}} \Phi^{'}_{T_1^N\rightarrow B_1^N},
\end{equation}
and 
\begin{equation}
\ket{\Psi^{''}_{j_{P_{\bar{w}}}}}_{A_{P_w} I_1^N M_1^N T_1^N}  =\sum_{j_{P_w}} \sqrt{\operatorname{Pr}\left[j_{P_w}\right]}\left|j_{P_w}\right\rangle_{A_{P_w}}\left|\Psi_{j_1^N}^{\prime}\right\rangle_{I_1^N M_1^N T_1^N}.
\end{equation}
The normalized state conditional on the outcome $j_{P_{\bar{w}}}$ can thus be written as
\begin{equation}
\label{eq:state_prime}
    {\rho}^{', j_{P_{\bar{w}}}}_{A_{P_w} I_{P_w} M_{P_w} B_{P_w}} = \Phi^{''}_{I_{P_{\bar{w}}} M_{P_{\bar{w}}} T_1^N\rightarrow B_{P_w}}\Big( \ket{\Psi^{''}_{j_{P_{\bar{w}}}}}\bra{\Psi^{''}_{j_{P_{\bar{w}}}}}_{A_{P_w} I_1^N M_1^N T_1^N}\Big).
\end{equation}
Now, consider the measurements performed by Alice and Bob on the rounds in $P_w$, through which they learn $\overrightarrow{n}^{w}$,  $e_{\text{ph},1}^{w}$ and $N_{{\rm key},1}^{{\rm det}, w}$. Again, we can define a simple two-outcome POVM $\left\{M_{A_{P_w} I_{P_w} M_{ P_w}B_{P_w}}^{1,w}, M_{A_{P_w} I_{P_w} M_{ P_w}B_{P_w}}^{2,w}\right\}$ that only checks whether $N_{{\rm key},1}^{{\rm det}, w}<\underline{N}_{{\rm key},1}^{{\rm det}, w}\cup{e}_{\text{ph},1}^{w} > g_{\epsilon}(\overrightarrow{n}^{w}, N^{w})$ or not. 
By doing so, we can write
\begin{equation}
\label{eq:not_equivalent}
\operatorname{Pr}\left[N_{{\rm key},1}^{{\rm det}, w}<\underline{N}_{{\rm key},1}^{{\rm det}, w}\cup e_{\text{ph},1}^{w}>g_{\epsilon}(\overrightarrow{{n}}^{w}, N^{w})|j_{P_{\bar{w}}}\right]=\operatorname{Tr}\left[{M}_{A_{P_w} I_{P_w} M_{P_w} B_{P_w}}^{1,w} \rho_{A_{P_w} I_{P_w} M_{P_w} B_{P_w}}^{', j_{P_{\bar{w}}}}\right] .
\end{equation}
We have to show that \cref{eq:not_equivalent} can be rewritten in such a way that it becomes equivalent to \cref{eq:target} from the uncorrelated protocol.

By assumption (see \cref{eq:isomety_existance}), there exists an isometry $V^{j_{P_{\bar{w}}}}_{T_{P_w} \rightarrow I_{P_{\bar{w}}} M_{P_{\bar{w}}} T_1^N}$ such that $\Big\{(V^{j_{P_{\bar{w}}}}_{T_{P_w} \rightarrow I_{P_{\bar{w}}} M_{P_{\bar{w}}} T_1^N})^{\dagger}  \left|\Psi_{j_1^N}^{\prime}\right\rangle_{I_1^N M_1^N T_1^N}\Big\}_{j_{P_w}\in \mathcal{J}^{N,w}}\in S^{N,w}$. Let us rewrite \cref{eq:state_prime} as,
\begin{equation}
\begin{aligned}
    &{\rho}^{', j_{P_{\bar{w}}}}_{A_{P_w} I_{P_w} M_{P_w} B_{P_w}}= \Phi^{''}_{I_{P_{\bar{w}}} M_{P_{\bar{w}}} T_1^N\rightarrow B_{P_w}}\Big( \ket{\Psi^{''}_{j_{P_{\bar{w}}}}}\bra{\Psi^{''}_{j_{P_{\bar{w}}}}}_{A_{P_w} I_1^N M_1^N T_1^N}\Big)\\
    &=\Phi^{''}_{I_{P_{\bar{w}}} M_{P_{\bar{w}}} T_1^N\rightarrow B_{P_w}}\Big( V^{j_{P_{\bar{w}}}}_{T_{P_w} \rightarrow I_{P_{\bar{w}}} M_{P_{\bar{w}}} T_1^N} (V^{j_{P_{\bar{w}}}}_{T_{P_w} \rightarrow I_{P_{\bar{w}}} M_{P_{\bar{w}}} T_1^N})^{\dagger}   \ket{\Psi^{''}_{j_{P_{\bar{w}}}}}\bra{\Psi^{''}_{j_{P_{\bar{w}}}}}_{A_{P_w} I_1^N M_1^N T_1^N} V^{j_{P_{\bar{w}}}}_{T_{P_w} \rightarrow I_{P_{\bar{w}}} M_{P_{\bar{w}}} T_1^N} (V^{j_{P_{\bar{w}}}}_{T_{P_w} \rightarrow I_{P_{\bar{w}}} M_{P_{\bar{w}}} T_1^N})^{\dagger}\Big)\\
    &=\Phi^{''}_{I_{P_{\bar{w}}} M_{P_{\bar{w}}} T_1^N\rightarrow B_{P_w}}\Big(V^{j_{P_{\bar{w}}}}_{T_{P_w} \rightarrow I_{P_{\bar{w}}} M_{P_{\bar{w}}} T_1^N}  \ket{\Psi^{'''}_{j_{P_{\bar{w}}}}}\bra{\Psi^{'''}_{j_{P_{\bar{w}}}}}_{A_{P_w} I_{P_w} M_{P_w} T_{P_w}}(V^{j_{P_{\bar{w}}}}_{T_{P_w} \rightarrow I_{P_{\bar{w}}} M_{P_{\bar{w}}} T_1^N})^{\dagger} \Big)\\
    &= \Phi^{'''}_{T_{P_w}\rightarrow B_{P_w}}\Big( \ket{\Psi^{'''}_{j_{P_{\bar{w}}}}}\bra{\Psi^{'''}_{j_{P_{\bar{w}}}}}_{A_{P_w} I_{P_w} M_{P_w} T_{P_w}}\Big),
\end{aligned}
\end{equation}
where 
\begin{equation}
\label{eq:triple_operator}
    \Phi^{'''}_{T_{P_w}\rightarrow B_{P_w}} (\sigma)= \Phi^{''}_{I_{P_{\bar{w}}} M_{P_{\bar{w}}} T_1^N\rightarrow B_{P_w}}\Big( V^{j_{P_{\bar{w}}}}_{T_{P_w} \rightarrow I_{P_{\bar{w}}} M_{P_{\bar{w}}} T_1^N}  \sigma(V^{j_{P_{\bar{w}}}}_{T_{P_w} \rightarrow I_{P_{\bar{w}}} M_{P_{\bar{w}}} T_1^N})^{\dagger}\Big),
\end{equation}
and 
\begin{equation}
\label{eq:triple}
\begin{aligned}
    \ket{\Psi^{'''}_{j_{P_{\bar{w}}}}}_{A_{P_w} I_{P_w} M_{P_w} T_{P_w}} &= (V^{j_{P_{\bar{w}}}}_{T_{P_w} \rightarrow I_{P_{\bar{w}}} M_{P_{\bar{w}}} T_1^N})^{\dagger}   \ket{\Psi^{''}_{j_{P_{\bar{w}}}}}_{A_{P_w} I_1^N M_1^N T_1^N} \\
    &=\sum_{j_{P_w}} \sqrt{\operatorname{Pr}\left[j_{P_w}\right]}\left|j_{P_w}\right\rangle_{A_{P_w}}(V^{j_{P_{\bar{w}}}}_{T_{P_w} \rightarrow I_{P_{\bar{w}}} M_{P_{\bar{w}}} T_1^N})^{\dagger}  \left|\Psi_{j_1^N}^{\prime}\right\rangle_{I_1^N M_1^N T_1^N}.
\end{aligned}
\end{equation}
Now we can rewrite \cref{eq:not_equivalent} as 
\begin{equation}
\label{eq:equivalent}
\begin{aligned}
&\operatorname{Pr}\left[N_{{\rm key},1}^{{\rm det}, w}<\underline{N}_{{\rm key},1}^{{\rm det}, w}\cup e_{\text{ph},1}^{w}>g_{\epsilon}(\overrightarrow{{n}}^{w}, N^{w})|j_{P_{\bar{w}}}\right]\\
=&\operatorname{Tr}\left[\hat{M}_{A_{P_w} I_{P_w} M_{P_w} B_{P_w}}^{1,w} \Phi^{'''}_{T_{P_w}\rightarrow B_{P_w}}\Big( \ket{\Psi^{'''}_{j_{P_{\bar{w}}}}}\bra{\Psi^{'''}_{j_{P_{\bar{w}}}}}_{A_{P_w} I_{P_w} M_{P_w} T_{P_w}}\Big)\right].
\end{aligned}
\end{equation}

Note that the state given by \cref{eq:triple} has precisely the form of the uncorrelated source-replacement state in \cref{eq:uncorrelated_state} for a protocol with $N^{w}$ rounds where the states $\Big\{\left|\Psi_{j_1^N}\right\rangle_{I_1^N M_1^N T_1^N}\Big\}_{j_{1}^{N}\in \mathcal{J}^{N}}\in S^{N}$ have been replaced by $\Big\{(V^{j_{P_{\bar{w}}}}_{T_{P_w} \rightarrow I_{P_{\bar{w}}} M_{P_{\bar{w}}} T_1^N})^{\dagger}  \left|\Psi_{j_1^N}^{\prime}\right\rangle_{I_1^N M_1^N T_1^N}\Big\}_{j_{P_w}\in \mathcal{J}^{N,w}}\in S^{N,w}.$ Because of this, \cref{eq:equivalent} has the same form as \cref{eq:target}, which by assumption is valid for any $N$ and holds for any CPTP map. Therefore, it follows that $\operatorname{Pr}\left[N_{{\rm key},1}^{{\rm det}, w}<\underline{N}_{{\rm key},1}^{{\rm det}, w}\cup e_{\text{ph},1}^{w}>g_{\epsilon}(\overrightarrow{{n}}^{w},N^{w})|j_{P_{\bar{w}}}\right]\leq \epsilon$ and by applying the law of total probability, we have that
\begin{equation}
\begin{aligned}
&\operatorname{Pr}\left[N_{{\rm key},1}^{{\rm det}, w}<\underline{N}_{{\rm key},1}^{{\rm det}, w}\cup e_{\text{ph},1}^{w}>g_{\epsilon}(\overrightarrow{{n}}^{w},N^{w})\right] = \\
&\sum_{j_{P_{\bar{w}}}}\operatorname{Pr}[j_{P_{\bar{w}}}]\operatorname{Pr}\left[N_{{\rm key},1}^{{\rm det}, w}<\underline{N}_{{\rm key},1}^{{\rm det}, w}\cup e_{\text{ph},1}^{w}>g_{\epsilon}(\overrightarrow{{n}}^{w},N^{w})|j_{P_{\bar{w}}}\right]\leq  \sum_{j_{P_{\bar{w}}}}\operatorname{Pr}[j_{P_{\bar{w}}}]\epsilon \leq\epsilon.
\end{aligned}
\end{equation}

Note that this result applies for the $w$-th phase-error estimation protocol, which we now write explicitly. However, we want to bound the phase-error rate for the original phase-error estimation protocol, and for now we have that
\begin{equation}
\operatorname*{Pr}_{w\text{-th PEEP}}
\left[ N_{{\rm key},1}^{{\rm det}, w}<\underline{N}_{{\rm key},1}^{{\rm det}, w}\cup e_{\text{ph},1}^{w} > g_{\epsilon}(\overrightarrow{n}^{w},N^{w}) \right]
\le \epsilon .
\end{equation}

For the next steps of the proof, we define the full phase-error estimation protocol (PEEP) below

\begin{tcolorbox}[breakable, colback=brown!1,
  colframe=black!30,title=Full phase-error estimation protocol (PEEP)]
\begin{enumerate}
\item[1--5.] Same as in Source-replaced protocol.
\item[6.] \textbf{Measurements in sifted-key rounds:} Partition the rounds $k \in\{1, \ldots, N\}$ into the sets $P_w=\{k: k \equiv w \left.\left(\bmod l_c+1\right)\right\}$. Then:

(a) Phase-error measurements for rounds in $P_w$ : For each round $k \in \mathcal{D}_{\text {key }} \cap P_w$, Alice measures $A_k$ in $\left\{|+\rangle_{A_k},|-\rangle_{A_k}\right\}$ and Bob measures the filtered system using $\left\{G_{+}, G_{-}\right\}$. We denote the single-photon phase-error rate of the $w$-th partition ${e}_{\text{ph},1}^{w}$ as the fraction of events in which their outcomes differ and $m_k=1$.
\end{enumerate}
\end{tcolorbox}

The statistics of the random variables $e_{\text{ph},1}^{w}$, $N_{{\rm key},1}^{{\rm det}, w}$ and $\overrightarrow{n}^{w}$ depend only on the marginal state of the systems $A_{P_w}, I_{P_w}, M_{P_w}$ and $B_{P_w}$ as well as on the measurements performed in the rounds in $P_w$. Importantly, the full PEEP and the $w$-th PEEP only differ in the measurements in $P_{\bar{w}}$, which do not affect the marginal state in $P_w$. This means that the $a$ $priori$ distribution of $e_{\text{ph},1}^{w}$, $N_{{\rm key},1}^{{\rm det}, w}$ and $\overrightarrow{n}^{w}$ must be identical in both scenarios and so 
\begin{equation}
\begin{aligned}
&\operatorname*{Pr}_{\text{Full PEEP}}
\left[
N_{{\rm key},1}^{{\rm det}, w}<\underline{N}_{{\rm key},1}^{{\rm det}, w}\cup e_{\text{ph},1}^{w} > g_{\epsilon}(\overrightarrow{n}^{w},N^{w})
\right] =\\
&\operatorname*{Pr}_{w\text{-th PEEP}}
\left[
N_{{\rm key},1}^{{\rm det}, w}<\underline{N}_{{\rm key},1}^{{\rm det}, w}\cup e_{\text{ph},1}^{w} > g_{\epsilon}(\overrightarrow{n}^{w},N^{w})
\right]
\le \epsilon. 
\end{aligned}
\end{equation}
Next, we define the overall phase-error rate as
\begin{equation}
\label{eq:def_eph_app}
e_{\text{ph},1} := \frac{\sum_w N_{\mathrm{key},1}^{\mathrm{det}, w} e_{\text{ph},1}^{w}}{N_{\mathrm{key},1}^{\mathrm{det}}}
\end{equation}
where $N_{\mathrm{key},1}^{\mathrm{det}}$ denotes the total number of key rounds given by 
\begin{equation}
    N_{\mathrm{key},1}^{\mathrm{det}}:= \sum_{w}N_{\mathrm{key},1}^{\mathrm{det},w}.
\end{equation}
\cref{eq:def_eph_app} can be upper bounded by
\begin{equation}
\label{eq:rea_1}
    e_{{\rm ph},(1)}\leq \underset{w}{\rm max}\{e^{w}_{{\rm ph},(1)}\}.
\end{equation}
Finally, consider the following event for the full PEEP,
\begin{equation}
\label{eq:rea_2}
    \mathcal{B}= \left\{N_{\mathrm{key},1}^{\mathrm{det}}<\sum_{w}\underline{N}_{\mathrm{key},1}^{\mathrm{det},w}\cup e_{\text{ph},1}>\max_{w}g_{\epsilon}(\overrightarrow{n}^{w},N^{w})\right\}.
\end{equation}
It follows that,
\begin{equation}
\label{eq:rea_3}
\mathcal B
\subseteq
\bigcup_{w=0}^{l_c}
\left\{
N_{\mathrm{key},1}^{\mathrm{det},w}<\underline{N}_{\mathrm{key},1}^{\mathrm{det},w}\cup e_{\text{ph},1}^{w} > g_{\epsilon}(\overrightarrow{n}^{w},N^{w})
\right\}.
\end{equation}
Applying the union bound, we obtain
\begin{equation}
\label{eq:rea_4}
\operatorname{Pr}[\mathcal B]
\le
\sum_{w=0}^{l_c}
\operatorname{Pr}\!\left[ N_{\mathrm{key},1}^{\mathrm{det},w}<\underline{N}_{\mathrm{key},1}^{\mathrm{det},w}\cup 
e_{\text{ph},1}^{w} > g_{\epsilon}(\vec n^{w},N^{w})
\right]\leq (l_c+1)\epsilon.
\end{equation}
Therefore, in the full PEEP, we have that 
\begin{equation}
\label{eq:phase_error_final}
     \operatorname{Pr}[N_{\mathrm{key},1}^{\mathrm{det}}<\sum_{w}\underline{N}_{\mathrm{key},1}^{\mathrm{det},w}\cup e_{\text{ph},1}>\max_{w}g_{\epsilon}(\overrightarrow{n}^{w},N^{w})]\leq (l_c+1)\epsilon.
\end{equation}
\end{proof}

\begin{remark}
Provided that a lower bound $\underline{N}_{{\rm key},1}^{\rm det}$ on the number of single-photon detected key rounds and an upper bound $\overline{N}_{{\rm ph},1}$  for the number of phase errors in the uncorrelated case are available, i.e., it holds for that case  
\begin{equation}
\operatorname{Pr}\left[N_{{\rm key},1}^{{\rm det}}<\underline{N}_{{\rm key},1}^{{\rm det}} \cup N_{{\rm ph},1}>\overline{N}_{{\rm ph},1}\right] \leq \epsilon,
\end{equation}
an alternative to \cref{eq:phase_error_final} can be derived by taking the weighted average, in similar fashion to what is done in \cite{curraslorenzo2026rigorousphaseerrorestimationsecurityframework}. For that we combine 
\begin{equation}
e_{{\rm ph},1}:=\frac{N_{{\rm ph},1
    }}{N_{{\rm key},1}^{\rm det}}
\end{equation}
with \cref{eq:def_eph_app} and we follow the same reasoning as in \cref{eq:rea_1,eq:rea_2,eq:rea_3,eq:rea_4,}, which yields
\begin{equation}
    \label{eq:phase_error_improved}
     \operatorname{Pr}\left[N_{\mathrm{key},1}^{\mathrm{det}}<\underline{N}_{\mathrm{key},1}^{\mathrm{det}}\cup e_{{\rm ph},1}> \frac{\sum_{w} \overline{N}_{{\rm ph},1}^{w}}{\underline{N}_{\mathrm{key},1}^{\mathrm{det}}}\right]\leq (l_c+1)\epsilon.
\end{equation}
where we have defined $\underline{N}_{\mathrm{key},1}^{\mathrm{det}} =\sum_{w}\underline{N}_{\mathrm{key},1}^{\mathrm{det},w}$.
\end{remark}

\begin{corollary}\label{coro:1}
Consider a prepare-and-measure decoy-state QKD protocol with an uncorrelated source, where Alice generates a purification $\left|\psi_{j_k}^{k}\right\rangle_{I_k M_k T_k}$ when choosing setting $j_k \in \mathcal{J}$ in round $k$. Suppose there exists an admissibility set $\mathcal{S}$ of state families indexed by $j \in \mathcal{J}$ such that if

\begin{equation}
\left\{\left|\psi_{j_k}^{k}\right\rangle_{I_k M_k T_k}\right\}_{j_k \in \mathcal{J}} \in \mathcal{S}, \quad \forall k ,
\end{equation}
then the phase-error rate bound in \cref{eq:phase_error_basic} holds.

Now consider the analogous protocol with a source exhibiting correlations in bit and basis settings up to length $l_c$. For any sequence of settings $j_{k-l_c}^{k+l_c}$ (interpreted with the boundary conventions $j_{k-l_c}^{k+l_c} \equiv j_{{\rm max}(k-l_c,1)}^{{\rm min}(k+l_c,N)}$), define the joint state emitted in rounds $k$ to $k+l_c$ as
\begin{equation}
|\Psi_{j_{k-l_c}^{k+l_c}}^{k}\rangle_{I_k^{k+l_c}M_k^{k+l_c}T_k^{k+l_c}}=\bigotimes_{l=k}^{k+l_c}|\psi_{j_{l-l_c}^l}^{l}\rangle_{I_l M_l T_l} .
\end{equation}

Suppose that, for every round $k$ and every fixed choice of past and future settings $(j_{k-l_c}^{k-1}, j_{k+1}^{k+l_c})$, there exists an isometry
\begin{equation}
V_{T_k \rightarrow I_{k+1}^{k+l_c}M_{k+1}^{k+l_c}T_{k}^{k+l_c}}^{\left(j_{k-l_c}^{k-1}, j_{k+1}^{k+l_c}\right)}: \mathcal{H}_{T_k} \rightarrow \mathcal{H}_{I_{k+1}^{k+l_c}M_{k+1}^{k+l_c}T_{k}^{k+l_c}},
\end{equation}
such that
\begin{equation}
\left\{\left(V_{T_k \rightarrow I_{k+1}^{k+l_c}M_{k+1}^{k+l_c}T_{k}^{k+l_c}}^{\left(j_{k-l_c}^{k-1}, j_{k+1}^{k+l_c}\right)}\right)^{\dagger}|\Psi_{j_{k-l_c}^{k+l_c}}^{k}\rangle_{I_{k}^{k+l_c}M_{k}^{k+l_c}T_k^{k+l_c}}\right\}_{j_k \in \mathcal{J}} \in \mathcal{S} .
\end{equation}

Then, partitioning the rounds $\{1, \ldots, N\}$ into ($l_c+1$) sets $P_w=\left\{k: k \equiv w \bmod \left(l_c+1\right)\right\}$ with $w=0, \ldots, l_c$, the phase-error rate in the correlated scenario satisfies \cref{eq:phase_error_final}.
\end{corollary}
\begin{proof} 
We need to verify that the conditions of \cref{thm:2} are satisfied. That is, we need to show that, for every $w\in\{0,1,...,l_c\}$, there exists an isometry such that
\begin{equation}
    \Big\{(V^{j_{P_{\bar{w}}}}_{T_{P_w} \rightarrow I_{P_{\bar{w}}} M_{P_{\bar{w}}} T_1^N})^{\dagger}  \left|\Psi_{j_1^N}^{\prime}\right\rangle_{I_1^N M_1^N T_1^N}\Big\}_{j_{P_w}\in \mathcal{J}^{N,w}}\in \mathcal{S}^{N,w}.
\end{equation}
For this, let us define,
\begin{equation}
\mathcal{S}^{N^{w}}=\left\{\left\{\bigotimes_{k \in P_w}\left|\varphi_{j_k}^{k}\right\rangle_{I_k M_k T_k}\right\}_{j_{P_w \in \mathcal{J}^{w}}}\Big|\left\{\left|\varphi_{j_k}^{k}\right\rangle_{I_k M_k T_k}\right\}_{j_k \in \mathcal{J}} \in \mathcal{S}, \forall k \in P_w\right\} .
\end{equation}

Now, let us rewrite the global correlated state. In particular, let $k_{\min }^{w}=\min P_w$ denote the smallest round index in $P_w$. Note that $k_{\min }^{w}=w$ for $w \in\left\{1, \ldots, l_c\right\}$ and $k_{\min }^{(0)}=l_c+1$. Since consecutive elements of $P_w$ differ by exactly $l_c+1$ positions, the blocks $\left\{k, k+1, \ldots, \min \left(k+l_c, N\right)\right\}$ for $k \in P_w$ partition the rounds $\{k_{\min }^{w}, \ldots, N\}$. We can therefore write the global emitted state as

\begin{equation}
\begin{aligned}
&\left|\Psi_{j_1^N}^{\prime}\right\rangle_{I_1^N M_1^N T_1^N}=\bigotimes_{k=1}^N\left|\psi_{j_{\max \left(1, k-l_c\right)}^{k}}^{k}\right\rangle_{I_k M_k T_k}\\
&=\left(\bigotimes_{k=1}^{k_{\min }^{w}-1}\left|\psi_{j_{\max \left(1, k-l_c\right)}^k}^{k}\right\rangle_{I_k M_k T_k}\right) \otimes\left(\bigotimes_{k \in P_w} \bigotimes_{m=k}^{\min \left(k+l_c, N\right)}\left|\psi_{j_{\max \left(1, m-l_c\right)}^m}^{m}\right\rangle_{I_m M_m T_m}\right) \\
&=\left(\bigotimes_{k=1}^{k_{\min }^{w}-1}\left|\psi_{j_{\max \left(1, k-l_c\right)}^k}^{k}\right\rangle_{I_k M_k T_k}\right) \otimes\left(\bigotimes_{k \in P_w}\left|\Psi_{j_{\max \left(1, k-l_c\right)}^{\min \left(k+l_c, N\right)}}^{k}\right\rangle_{I_k^{\min \left(k+l_c, N\right)} M_k^{\min \left(k+l_c, N\right)} T_k^{\min \left(k+l_c, N\right)}}\right) .
\end{aligned}
\end{equation}

Let us now define the global isometry for group  $P_w$ conditional on some outcomes of the complementary rounds $j_{P_{\bar{w}}}$,
\begin{equation}
V_{T_{P_w} \rightarrow I_{P_{\bar{w}}}M_{P_{\bar{w}}}T_1^N}^{j_{P_{\bar{w}}}}=\left(\bigotimes_{k=1}^{k_{\min }^{w}-1}\left|\psi_{j_{\max \left(1, k-l_c\right)}^k}^{k}\right\rangle_{I_k M_k T_k}\right) \otimes\left(\bigotimes_{k \in P_w} V_{T_{k} \rightarrow I_{k+1}^{\min \left(k+l_c, N\right)} M_{k+1}^{\min \left(k+l_c, N\right)} T_k^{\min \left(k+l_c, N\right)}}^{\left(j_{\max (1, k-l_c) }^{k-1}, j_{k+1}^{\min \left(k+l_c, N\right)}\right)}\right) .
\end{equation}

Note that for each $k\in P_w$, the settings $\left(j_{\max (1, k-l_c) }^{k-1}, j_{k+1}^{\min \left(k+l_c, N\right)}\right)$ lie in ${P_{\bar{w}}}$. Furthermore, for each $k\in\{1,...,k_{\min }^{w}-1\}$ all indices $\{{\max \left(1, k-l_c\right),...,k\}}$ are strictly less than $k_{\min }^{w}$ and hence belong to the complementary set $P_{\bar{w}}$, while the rounds $\{k+1,...,k+l_c\}$, also belong to the set $P_{\bar{w}}$. Now, if we apply the  adjoint of the isometry to the global emitted state, we obtain
\begin{equation}
\begin{aligned}
&\left(V_{T_{P_w} \rightarrow I_{P_{\bar{w}}}M_{P_{\bar{w}}}T_1^N}^{j_{P_{\bar{w}}}}\right)^{\dagger}\left|\Psi_{j_1^N}^{\prime}\right\rangle_{I_1^N M_1^N T_1^N}=\\
&\bigotimes_{k \in P_w}\left( V_{T_{k} \rightarrow I_{k+1}^{\min \left(k+l_c, N\right)} M_{k+1}^{\min \left(k+l_c, N\right)} T_k^{\min \left(k+l_c, N\right)}}^{\left(j_{\max (1, k-l_c) }^{k-1}, j_{k+1}^{\min \left(k+l_c, N\right)}\right)}\right)^{\dagger}\left|\Psi_{j_{\max \left(1, k-l_c\right)}^{\min \left(k+l_c, N\right)}}^{k}\right\rangle_{I_k^{\min \left(k+l_c, N\right)} M_k^{\min \left(k+l_c, N\right)} T_k^{\min \left(k+l_c, N\right)}}.
\end{aligned}
\end{equation}
Finally, by assumption, for each $k\in P_w$, we have that
\begin{equation}
\left\{\left( V_{T_{k} \rightarrow I_{k+1}^{\min \left(k+l_c, N\right)} M_{k+1}^{\min \left(k+l_c, N\right)} T_k^{\min \left(k+l_c, N\right)}}^{\left(j_{\max (1, k-l_c) }^{k-1}, j_{k+1}^{\min \left(k+l_c, N\right)}\right)}\right)^{\dagger}\left|\Psi_{j_{\max \left(1, k-l_c\right)}^{\min \left(k+l_c, N\right)}}^{k}\right\rangle_{I_k^{\min \left(k+l_c, N\right)} M_k^{\min \left(k+l_c, N\right)} T_k^{\min \left(k+l_c, N\right)}}\right\}_{j_k\in\mathcal{J}}\in \mathcal{S},
\end{equation}
and so it follows that $\left\{\left(V_{T_{P_w} \rightarrow I_{P_{\bar{w}}}M_{P_{\bar{w}}}T_1^N}^{j_{P_{\bar{w}}}}\right)^{\dagger}\left|\Psi_{j_1^N}^{\prime}\right\rangle_{I_1^N M_1^N T_1^N}\right\}_{j_{P_w}\in\mathcal{J}^{N,w}}\in \mathcal{S}^{N,w}$, which concludes the proof.
\end{proof}

\begin{corollary}\label{coro:2}
Consider a decoy-state prepare-and-measure QKD protocol with an uncorrelated source whose state for round $k$ is defined as
\begin{equation}
\ket{\Psi_{j_k}^{k}}_{I_k M_k T_k}=\sum_{\mu_k}\sqrt{p_{\mu_k}}\ket{{\mu}_k}_{I_k}\sum_{m_k=0}^{\infty} \sqrt{p_{m_k|\mu_k}}\ket{m_k}_{M_k}\ket{\psi^{m_k,k}_{j_k}}_{T_k}.
\end{equation}
Suppose there exist a set of reference states $\{|\phi^{1}_j\rangle_T\}_{j\in\mathcal{J}}$ such that the phase-error rate bound in \cref{eq:phase_error_basic} holds as long as the single-photon component of the phase-randomized weak coherent pulse satisfies 
\begin{equation}
\Big|\langle\phi^{m_k=1}_{j_k} \mid \psi_{j_k}^{m_k=1,k}\rangle_{T_k}\Big|^2 \geq 1-\epsilon_{j_k}, \quad \forall k, \forall j_k \in \mathcal{J}.
\end{equation}

Let us now consider an analogous protocol with a source exhibiting correlations in the bit and basis settings up to length $l_c$. Suppose that for every round $k$ and every choice of past and future settings $(j_{k-l_c}^{{k-1}},j_{k+1}^{{k+l_c}})$, there exist a family of states $\{|{\phi}^{{m_k=1, }k}_{j^{k+l_c}_{k-l_c}}\rangle_{T_k}\}_{j_k\in\mathcal{J}}$ with the same Gram matrix as $\{|{\phi}^{{m_k=1}}_{j_k}\rangle_{ T_k}\}_{j_k\in\mathcal{J}}$. Additionally, assume that there exist some states, independent of $j_k$, $|\lambda_{j_{k-l_c+1}^{k-1}, j_{k+1}^{k+l_c}}^{k}\rangle_{I_{k+1}^{k+l_c}M_{k+1}^{k+l_c}T_{k+1}^{k+l_c}}$ such that,

\begin{equation}\label{eq:inner}
\mid\langle{\phi}_{j_{k-l_c}^{k+l_c}}^{{m_k=1},k}|_{T_k} \otimes\langle\lambda_{j_{k-l_c+1}^{k-1}, j_{k+1}^{k+l_c}}^{k}|_{I_{k+1}^{k+l_c}M_{k+1}^{k+l_c}T_{k+1}^{k+l_c}} \mid \Psi_{j_{k-l_c}^{k+l_c}}^{m_{k}=1,k}\rangle_{I_{k+1}^{k+l_c}M_{k+1}^{k+l_c}T_k^{k+l_c}}|^2 \geq 1-\epsilon_{j_k}, \quad \forall j_k.
\end{equation}

Then, the phase-error rate bound in \cref{eq:phase_error_final} holds for the correlated scenario.
\end{corollary}

\begin{proof}
We want to check that the conditions in \cref{coro:1} hold. Precisely, we need to prove that for any fixed $(j^{k-1}_{k-l_c},j^{k+l_c}_{k+1})$ there exists an isometry such that \footnote{Here, for simplicity of notation, we remove the boundary conditions},
\begin{equation} \label{eq:S_condition}
\left\{\left(V_{T_k \rightarrow I_{k+1}^{k+l_c}M_{k+1}^{k+l_c}T_{k}^{k+l_c}}^{\left(j_{k-l_c}^{k-1}, j_{k+1}^{k+l_c}\right)}\right)^{\dagger}|\Psi_{j_{k-l_c}^{k+l_c}}^{k}\rangle_{I_{k}^{k+l_c}M_{k}^{k+l_c}T_k^{k+l_c}}\right\}_{j_k \in \mathcal{J}} \in \mathcal{S} .
\end{equation}
The proof for the uncorrelated scenario, defines the per-round admissibility set. Let $\ket{\tilde{\varphi}_{j}}_{I M T}=\sum_{\mu}\sqrt{p_{\mu}}\ket{{\mu}}_{I}\sum_{m=0}^{\infty} \sqrt{p_{m|\mu}}\ket{m}_{M}\ket{\varphi^{m}_{j}}_{T}$, then we have that such set is defined as
\begin{equation} \label{eq:S}
\mathcal{S}=\left\{\left\{\left|\tilde{\varphi}_{j}\right\rangle_{I M T}\right\}_{j\in\mathcal{J}}\Big||\langle\phi^{1}_j|{\varphi}^{1}_j\rangle_T|^2\geq1-\epsilon_j, \quad \forall j\in\mathcal{J}\right\}.
\end{equation}

Now, by assumption, the states $\{|{\phi}^{{m_k=1, }k}_{j^{k+l_c}_{k-l_c}}\rangle_{T_k}\}_{j_k\in\mathcal{J}}$ have the same Gram matrix as $\{|{\phi}^{{m_k=1}}_{j_k}\rangle_{T_k}\}_{j_k\in\mathcal{J}}$. This means that there exists an isometry $V^{(1)}_{T_k\rightarrow T_k}$ such that $V^{(1)}_{T_k\rightarrow T_k}|{\phi}^{{m_k=1}}_{j_k}\rangle_{T_k} = |{\phi}^{{m_k=1, }k}_{j^{k+l_c}_{k-l_c}}\rangle_{ T_k} $. 
Moreover, for any normalized state $|\lambda_{j_{k-l_c+1}^{k-1}, j_{k+1}^{k+l_c}}^{k} \rangle_{I_{k+1}^{k+l_c}M_{k+1}^{k+l_c}T_{k+1}^{k+l_c}}$, there exists an isometry $V^{(2)}_{T_k\rightarrow I_{k+1}^{k+l_c} M_{k+1}^{k+l_c} T_{k}^{k+l_c}}$ such that $V^{(2)}_{T_k\rightarrow I_{k+1}^{k+l_c} M_{k+1}^{k+l_c} T_{k}^{k+l_c}}\ket{\cdot}_{T_k}= \ket{\cdot}_{T_k}|\lambda_{j_{k-l_c+1}^{k-1}, j_{k+1}^{k+l_c}}^{k} \rangle_{I_{k+1}^{k+l_c}M_{k+1}^{k+l_c}T_{k+1}^{k+l_c}}$.  Then, it follows that
\begin{equation}
\begin{aligned}
&V^{(2)}_{T_k\rightarrow I_{k+1}^{k+l_c} M_{k+1}^{k+l_c} T_{k}^{k+l_c}}  \circ V^{(1)}_{T_k\rightarrow T_k}|{\phi}_{j_k}^{{m_k=1}}\rangle_{T_k} = |{\phi}_{j^{k+l_c}_{k-l_c}}^{{m_k=1},k}\rangle_{T_k} |\lambda_{j_{k-l_c+1}^{k-1}, j_{k+1}^{k+l_c}}^{k} \rangle_{I_{k+1}^{k+l_c}M_{k+1}^{k+l_c}T_{k+1}^{k+l_c}}.\\
\end{aligned}
\end{equation}
Note that $V^{(2)}_{T_k\rightarrow I_{k+1}^{k+l_c} M_{k+1}^{k+l_c} T_{k}^{k+l_c}}  \circ V^{(1)}_{T_k\rightarrow T_k}$ depends on $(j_{k-l_c}^{k-1}, j_{k+1}^{k+l_c})$, and so the composition of both isometries has the form required by \cref{eq:S_condition}, and thus can be expressed as $V^{(j_{k-l_c}^{k-1}, j_{k+1}^{k+l_c})}_{T_k\rightarrow I_{k+1}^{k+l_c} M_{k+1}^{k+l_c} T_{k}^{k+l_c}} = V^{(2)}_{T_k\rightarrow I_{k+1}^{k+l_c} M_{k+1}^{k+l_c} T_{k}^{k+l_c}}  \circ V^{(1)}_{T_k\rightarrow T_k}$.

Now, if we apply the hermitian transpose of this isometry to the state $|\Psi_{j_{k-l_c}^{k+l_c}}^{k}\rangle_{I_{k}^{k+l_c}M_{k}^{k+l_c}T_k^{k+l_c}}$, it follows that,
\begin{equation}
\begin{aligned}\label{eq:new_state}
    &\left( V^{(j_{k-l_c}^{k-1}, j_{k+1}^{k+l_c})}_{T_k\rightarrow I_{k+1}^{k+l_c} M_{k+1}^{k+l_c} T_{k}^{k+l_c}}\right)^{\dagger} |\Psi_{j_{k-l_c}^{k+l_c}}^{k}\rangle_{I_{k}^{k+l_c}M_{k}^{k+l_c}T_k^{k+l_c}} \\ &=\sum_{\mu_k}\sqrt{p_{\mu_k}}\ket{{\mu_k}}_{I_k}\sum_{m_k=0}^{\infty} \sqrt{p_{m_k|\mu_k}}\ket{m_k}_{M_k} \left( V^{(j_{k-l_c}^{k-1}, j_{k+1}^{k+l_c})}_{T_k\rightarrow I_{k+1}^{k+l_c} M_{k+1}^{k+l_c} T_{k}^{k+l_c}}\right)^{\dagger} |\Psi_{j_{k-l_c}^{k+l_c}}^{m_k,k}\rangle_{I_{k+1}^{k+l_c}M_{k+1}^{k+l_c}T_k^{k+l_c}},\\
\end{aligned}
\end{equation}
where we have defined $|\Psi_{j_{k-l_c}^{k+l_c}}^{m_k,k}\rangle_{I_{k+1}^{k+l_c}M_{k+1}^{k+l_c}T_k^{k+l_c}} =\ket{\psi^{m_k,k}_{j^k_{k-l_c}}}_{T_k} \bigotimes_{d=k+1}^{k+l_c}\left|\psi_{j_{ d-l_c}^d}^{d}\right\rangle_{I_d M_d T_d}$.
Our goal is to determine whether the state in \cref{eq:new_state} belongs to $\mathcal{S}$. To this end, note that for a state to belong to 
$\mathcal{S}$, it must satisfy the condition in \cref{eq:S}, which for that state reads,
\begin{equation}
\begin{aligned}
   &\Big| \langle{\phi}_{j_k}^{m_k=1}|_{T_k}(V^{(j_{k-l_c}^{k-1}, j_{k+1}^{k+l_c})}_{T_k\rightarrow I_{k+1}^{k+l_c} M_{k+1}^{k+l_c} T_{k}^{k+l_c}})^{\dagger} |\Psi_{j_{k-l_c}^{k+l_c}}^{m_k=1,k}\rangle_{I_{k+1}^{k+l_c}M_{k+1}^{k+l_c}T_k^{k+l_c}} \Big|^2=\\
   & \Big|\langle\phi_{j^{k+l_c}_{k-l_c}}^{m_k=1,k}|_{T_k} \langle\lambda_{j_{k-l_c+1}^{k-1}, j_{k+1}^{k+l_c}}^{k} |_{I_{k+1}^{k+l_c}M_{k+1}^{k+l_c}T_{k+1}^{k+l_c}}|\Psi_{j_{k-l_c}^{k+l_c}}^{m_k=1,k}\rangle_{I_{k+1}^{k+l_c}M_{k+1}^{k+l_c}T_k^{k+l_c}} \Big|^2\geq 1-\epsilon_{j_k},
\end{aligned}
\end{equation}
where in the last inequality we have used the assumption highlighted in \cref{eq:inner}. This means that the family of states $\left\{\left(V_{T_k \rightarrow I_{k+1}^{k+l_c}M_{k+1}^{k+l_c}T_{k}^{k+l_c}}^{\left(j_{k-l_c}^{k-1}, j_{k+1}^{k+l_c}\right)}\right)^{\dagger}|\Psi_{j_{k-l_c}^{k+l_c}}^{k}\rangle_{I_{k}^{k+l_c}M_{k}^{k+l_c}T_k^{k+l_c}}\right\}_{j_k \in \mathcal{J}} \in \mathcal{S},$ as we wanted to prove.
\end{proof}

\section{Concentration inequalities}\label{sec:concentration}

Here, we present the concentration inequalities used in \cref{sec:decoy}. For a more detailed discussion and the corresponding proofs of \cref{th1,th2} below, the reader is referred to \cite{Mannalath_decoy}.

\begin{theorem} \label{th1} {\bf\cite{Mannalath_decoy}}
Let $q$ be the average of $n$ independent Bernoulli variables, with expected value $\mathbb{E}[q]$. Then, for all $\epsilon \in(0,1 / 4]$, we have that $\operatorname{Pr}\left[\mathbb{E}[q]<\mathcal{F}_{n, \epsilon}^{\mathrm{L}}(q)\right] \leq \epsilon$ for

\begin{equation}
\begin{aligned}
\label{eq:F_def}
\mathcal{F}_{n, \epsilon}^{\mathrm{L}}(q)= 
\begin{cases}0 & \text { if } n q=0, \\ q-\frac{1-\epsilon}{n\left[1-\epsilon^*(n q, n)\right]} & \text { if } n q>0, \epsilon^*(n q, n) \leq \epsilon \leq 1, \\ I_\epsilon^{-1}(n q, n-n q+1) & \text { if } n q>0,0 \leq \epsilon \leq \epsilon^*(n q, n),\end{cases}
\end{aligned}
\end{equation}
where $\epsilon^*(n q, n)=I_{(n q-1) / n}(n q, n-n q+1)$ and $I_\epsilon^{-1}(a, b)$ denotes the inverse regularized incomplete beta function, such that $I_{I_\epsilon^{-1}(a, b)}(a, b)=\epsilon$ for
\begin{equation}
I_x(a, b)=\frac{\int_0^x t^{a-1}(1-t)^{b-1} d t}{\int_0^1 t^{a-1}(1-t)^{b-1} d t}.
\end{equation}

Complementarily, for all $\epsilon \in(0,1 / 4]$, we have that $\operatorname{Pr}\left[\mathbb{E}[q]>\mathcal{F}_{n, \epsilon}^{\mathrm{U}}(q)\right] \leq \epsilon$ for $\mathcal{F}_{n, \epsilon}^{\mathrm{U}}(q)=1-\mathcal{F}_{n, \epsilon}^{\mathrm{L}}(1-q)$.
\end{theorem}

\begin{theorem}\label{th2} {\bf \cite{Mannalath_decoy}}
Let $q$ be the average of $n$ independent and identically distributed Bernoulli variables, with expected value $\mathbb{E}[q]$. Then for all $\epsilon>0$ we have that
\begin{equation}
\operatorname{Pr}\left[q \geq \mathcal{G}_{n, \epsilon}^{\mathrm{U}}(\mathbb{E}[q])\right] \leq \epsilon,
\end{equation}
for
\begin{equation}
\label{eq:G0}
\mathcal{G}_{n, \epsilon}^{\mathrm{U}}(\mathbb{E}[q])=\frac{1}{n}\min \left\{k\in \{-1,...,n+1\} :  I_{\mathbb{E}[q]}(k, n-k+1) \leq \epsilon\right\},
\end{equation}
where $I_{\mathbb{E}[q]}(k,n-k+1)=\sum_{j=k}^{n}{n\choose j}{\mathbb{E}[q]}^j(1-\mathbb{E}[q])^{n-j}$.

Similarly we have that,
\begin{equation}
\operatorname{Pr}\left[q \leq \mathcal{G}_{n, \epsilon}^{\mathrm{L}}(\mathbb{E}[q])\right] \leq \epsilon,
\end{equation}
where $\mathcal{G}_{n, \epsilon}^{\mathrm{L}}(x)$ is defined as,
\begin{equation}
\label{eq:G1}
\mathcal{G}_{n, \epsilon}^{\mathrm{L}}(x)=\frac{1}{n}\max \left\{k\in\{-1,...,n+1\}: I_{1-\mathbb{E}[q]}(n-k, k+1) \leq \epsilon\right\}.
\end{equation}
\end{theorem}
Note that the quantities $G^{\rm U}_{n,\epsilon}$ and $G^{\rm L}_{n,\epsilon}$ respectively correspond (up to a one-unit shift and a rescaling by $n$) to the $(1-\epsilon)$-quantile and the $\epsilon$-quantile of a binomial distribution of parameters $n$ and $\mathbb{E}[q]$.

\begin{theorem}\label{th:hoef}
{\bf (Hoeffding’s inequality \cite{Hoeff})} Let $Q_1, Q_2, \ldots, Q_n$ represent a sequence of independent Bernoulli random variables and let $Q:=\sum_{i=1}^n Q_i$ with expectation value $\mathbb{E}[Q]$. Also, consider the standard notation $L \underset{\epsilon}{\leq}R$ indicating that the event $L \leq R$ occurs except with probability at most $\epsilon$, i.e., $\operatorname{Pr}[L>R] \leq \epsilon$. Then, we have that
\begin{equation}
    Q-\Delta_H(n) \underset{\epsilon_H}{\leq} \mathbb{E}[Q] \underset{\epsilon_H}{\leq} Q+\Delta_H(n),
\end{equation}
where $\Delta_H(n)=\sqrt{n \ln \left(1 / \epsilon_H\right) / 2}$.
\end{theorem}

\begin{theorem}
\label{th:azuma}
{\bf (Azuma’s inequality \cite{Azuma,Guille_framework})}  Consider $Q_1, Q_2, \ldots, Q_n$ to be any sequence of Bernoulli random variables, not necessarily independent, and let $\left\{\mathcal{F}_l\right\}_{l=1}^n$ be a filtration such that $Q_l$ is $\mathcal{F}_{l}$ measurable for all $l$.
Also, let $\Lambda_l=\sum_{u=1}^l Q_u$, with $l \leq n$. 
This bound states that
\begin{equation}
\begin{aligned}
& \operatorname{Pr}\left(\Lambda_n-\sum_{u=1}^n \operatorname{Pr}\left(Q_u=1 \mid \mathcal{F}_{u-1}\right)>b \sqrt{n}\right) \leq \exp \left(-\frac{b^2}{2}\right), \text { and } \\
& \operatorname{Pr}\left(\sum_{u=1}^n \operatorname{Pr}\left(Q_u=1 \mid \mathcal{F}_{u-1}\right)-\Lambda_n>b \sqrt{n}\right) \leq \exp \left(-\frac{b^2}{2}\right),
\end{aligned}
\end{equation}
for any $b>0$. In particular, this means that
\begin{equation}
\Lambda_n-\Delta_A(n) \underset{\epsilon_A}{\leq} \sum_{u=1}^n \operatorname{Pr}\left(Q_u=1 \mid \mathcal{F}_{u-1}\right) \underset{\epsilon_A}{\leq} \Lambda_n+\Delta_A(n),
\end{equation}
where $\Delta_A(n)=\sqrt{2 n \ln \left(1 / \epsilon_A\right)}$.
\end{theorem}

\section{Generalization of \cite{misma} to decoy-state protocols }\label{app:detector}

Here, we show how to generalize the main result of \cite{misma} to decoy-state protocols, thereby enabling the incorporation of detector imperfections within a phase-error-rate-based security proof. Although \cite{misma} already considers this type of protocols, applying its results requires solving a highly demanding optimization problem, which may be computationally unfeasible in practice. Crucially, if additional monotonicity properties of the functions $\mathcal{E}^{\rm decoy}_{0,0}$ and $\mathcal{F}^{\rm decoy}_{0,0}$ (see \cite{misma} for its precise definitions) could be established, this optimization could be removed. In this Appendix, we extend the approach of \cite{misma} to the decoy-state setting and show that monotonicity arguments involving only $\mathcal{F}^{\rm decoy}_{0,0}$ are sufficient to remove the need for such costly optimization procedures.

For generality, we employ once again the notation in Appendix~\ref{app:rigurous}, so that the number of single-photon detected key rounds is denoted as $N^{\rm det}_{{\rm key},1}$ with the understanding that in the protocol considered in the main text, these rounds correspond to $N^{{\rm det},Z}_{Z,1}$.  

\begin{theorem}[Variation of Theorem 2 from \cite{misma}]\label{thm:maintheoremdecoy}
  Let $\Sindep$ be the phase-error-estimation protocol of a particular decoy-state scenario when Bob's measurement setup satisfies the basis-independent detection efficiency condition. Let $\Sdep$ be the phase-error estimation protocol of the same decoy-state scenario when Bob's measurement setup suffers from a detection efficiency mismatch parameterized by $\delta_1$ and $\delta_2$, as defined in \cite{misma}. Consider the scenario in which Bob makes an active decision about whether to run $\Sindep$ or $\Sdep$. Let $\Prindep$ ($\Prdep$) be the probability measure conditional on the choice of $\Sindep$ ($\Sdep$). Suppose that  
\begin{equation}
\label{eq:guarantee_nK1_decoy}
\begin{aligned}
\Prindep \left[ N^{\rm det}_{{\rm key},1} < \mathcal{M}_L (\overrightarrow{N}_{Z}) \right]\leq \epsspone^2, \\
\Prindep \left[N^{\rm det}_{{\rm key},1} > \mathcal{M}_U (\overrightarrow{N}_{Z}) \right]\leq \epssptwo^2, \\
\end{aligned}
\end{equation}
where $\overrightarrow{N}_{Z (X)}$ represents all the announced statistics associated to those events in which Bob selects the $Z(X)$ basis, $\mathcal{M}_{L(U)}$ is a function that provides suitable lower (upper) bounds and 
\begin{equation}
\label{eq2:guarantee_BIDE_decoy}
\Prindep( e_{{\rm ph},1}> \Eindep^\mathrm{decoy}({\overrightarrow{N}_X},N^{\rm det}_{{\rm key},1})) \leq \epssrc^2.
\end{equation} 
Also, let 
\begin{equation} \label{eq:F00_def}
\Findepdecoy (\overrightarrow{N}_{X},N^{\rm det}_{{\rm key},1}) := N^{\rm det}_{{\rm key},1} \Eindepdecoy (\overrightarrow{N}_{X}, N^{\rm det}_{{\rm key},1}).
\end{equation}
Then, for any $\epsATb$ and $\epsATc$, we have that

\begin{equation} \label{eq:objective_combined}
\begin{aligned}
   &\Prdep \bigg( N^{\rm det}_{{\rm key},1}< \mathcal{M}_L (\overrightarrow{N}_{Z}) \\
   &\bigcup {e_{\mathrm{ph},1}} > \max_{\mathcal{M}_L (\overrightarrow{N}_{Z})\leq \hat{n}\leq\mathcal{M}_U (\overrightarrow{N}_{Z})}\frac{\max_{n \in \mathcal{W}_{\delta_2}(\hat{n})}   \Findepdecoy (\overrightarrow{N}_{X},n)}{\hat{n}}  + \frac{ \delta_1 + \gamma^{\epsATb}_{\mathrm{bin}}(\mathcal{M}_L (\overrightarrow{N}_{Z}), \deltaone)}{1-\delta_2-\gamma^{\epsATc}_{\mathrm{bin}}(\mathcal{M}_L (\overrightarrow{N}_{Z}), \delta_2)} \bigg)\\
   &\leq \epsspone^2 + \epssptwo^2 + \epssrc^2 + \epsATb^2 + \epsATc^2,
\end{aligned}
\end{equation}
where $\gamma^{\varepsilon}_\mathrm{bin}$ is a finite-size deviation term defined as
\footnote{Equivalence (1) follows from the monotonicity of the binomial tail in its 
lower summation limit, together with the fact that, for integer $k^{\star} 
= n\,\mathcal{G}^{\mathrm{U}}_{n,\varepsilon^{2}}(\delta)$, the condition 
$\lfloor n(\delta+x)\rfloor \geq k^{\star}$ is equivalent to 
$n(\delta+x) \geq k^{\star}$.} 

\begin{equation}
\label{eq:gamma_bin}
    \gamma_{\mathrm{bin}}^{\varepsilon}(n, \delta):=   \max \Big\{0,\mathcal{G}_{n, \varepsilon^2}^{\mathrm{U}}(\delta)-\delta\Big\} \underset{(1)}{\equiv} \min \left\{x \geq 0: \sum_{i=\lfloor n(\delta+x)\rfloor}^n\binom{n}{i} \delta^i(1-\delta)^{n-i} \leq \varepsilon^2\right\},
\end{equation}
where $\mathcal{G}_{n, \varepsilon}^{\mathrm{U}}(\cdot)$ is defined in \cref{eq:G0}
and $\mathcal{W}_{\delta_2}m:=\mathcal{N}_{\leq}\left(\left\lfloor\frac{m}{1-\delta_2-\gamma_{\text {bin }}^{\varepsilon_{{\rm dep }-2}}\left(m, \delta_2\right)}\right\rfloor\right)$ with $\mathcal{N}_{\leq}m:=\{0,1, \ldots, m\}$.
\end{theorem}

\begin{proof}
The proof follows exactly the same structure as that of Theorem 2 in \cite{misma}, with a modification to the maximization range. Specifically, in the proof of \cite{misma}, the bound $N^{\rm det}_{{\rm key},1}\leq N^{{\rm det},Z}_{Z}$ , restricts the range of the maximization over $\hat{n}$. This is done because only the first bound in \cref{eq:guarantee_nK1_decoy} is presented in the statement of the theorem. In our case, this argument is replaced by the more immediate observation that
\begin{equation}
    N^{\rm det}_{{\rm key},1}\leq \mathcal{M}_{U}(\overrightarrow{N}_{Z}),
\end{equation}
except with probability $\epssptwo^2$ according to \cref{eq:guarantee_nK1_decoy}. 

The remainder of the proof follows from that of Theorem 2 in \cite{misma}. 
\end{proof}

Next, before stating the result used in the main text, we show that the monotonicity of $\Findepdecoy(\overrightarrow{N}_{X},N^{\rm det}_{{\rm key},1})$ with respect to $N^{\rm det}_{{\rm key},1}$ is typically satisfied in QKD security proofs, since the concentration inequalities involved generally converge sublinearly.

\begin{corollary}
Consider the same setting as in \cref{thm:maintheoremdecoy}, and assume that the function $\Findepdecoy(\overrightarrow{N}_{X},N^{\rm det}_{{\rm key},1})$ is non-decreasing with respect to $N^{\rm det}_{{\rm key},1}$ \footnote{As it is shown in \cite{misma}, this is the case if the concentration inequalities used in the security proof  converge sublinearly}. Then, it follows that
\begin{equation} \label{eq:objective_combined_3}
\begin{aligned}
   &\Prdep \bigg( N^{\rm det}_{{\rm key},1} < \mathcal{M}_L (\overrightarrow{N}_{Z}) \bigcup  {e_{\mathrm{ph},1}} > \frac{ \Findepdecoy \Big(\overrightarrow{N}_{X},   \frac{\mathcal{M}_U (\overrightarrow{N}_{Z})}{1-\delta_2-\gamma^{\epsATc}_{\mathrm{bin}}(\mathcal{M}_L (\overrightarrow{N}_{Z}), \delta_2)}\Big)}{\mathcal{M}_L (\overrightarrow{N}_{Z})}  + \frac{ \delta_1 + \gamma^{\epsATb}_{\mathrm{bin}}(\mathcal{M}_L (\overrightarrow{N}_{Z}), \deltaone)}{1-\delta_2-\gamma^{\epsATc}_{\mathrm{bin}}(\mathcal{M}_L (\overrightarrow{N}_{Z}), \delta_2)} \bigg)\\
   &\leq \epsspone^2 + \epssptwo^2 + \epssrc^2 + \epsATb^2 + \epsATc^2,
\end{aligned}
\end{equation}
which represents a slightly looser bound than that of \cref{eq:objective_combined}, but does not require maximization.
\end{corollary}

\begin{proof}
The proof proceeds through the sequence of inequalities presented below. Our goal is to show that the left hand side of \cref{eq:objective_combined_3} is less or equal than that of \cref{eq:objective_combined}.
\begin{equation} \label{eq:objective_combined_2}
\begin{aligned}
    &\Prdep \bigg({e_{\mathrm{ph},1}} > \frac{ \Findepdecoy \Big(\overrightarrow{N}_{X},   \frac{\mathcal{M}_U (\overrightarrow{N}_{Z})}{1-\delta_2-\gamma^{\epsATc}_{\mathrm{bin}}(\mathcal{M}_L (\overrightarrow{N}_{Z}), \delta_2)}\Big)}{\mathcal{M}_L (\overrightarrow{N}_{Z})}  + \frac{ \delta_1 + \gamma^{\epsATb}_{\mathrm{bin}}(\mathcal{M}_L (\overrightarrow{N}_{Z}), \deltaone)}{1-\delta_2-\gamma^{\epsATc}_{\mathrm{bin}}(\mathcal{M}_L (\overrightarrow{N}_{Z}), \delta_2)} \bigg) \\
    &\leq \Prdep \bigg({e_{\mathrm{ph},1}} > \frac{ \Findepdecoy \Big(\overrightarrow{N}_{X},   \frac{\underset{{\mathcal{M}_L (\overrightarrow{N}_{Z})\leq \hat{n}\leq\mathcal{M}_U (\overrightarrow{N}_{Z})}}{\max} \hat n}{1-\delta_2-\underset{{\mathcal{M}_L (\overrightarrow{N}_{Z})\leq \hat{n}\leq\mathcal{M}_U (\overrightarrow{N}_{Z})}}{\min}\gamma^{\epsATc}_{\mathrm{bin}}(\hat n, \delta_2)}\Big)}{\underset{\mathcal{M}_L (\overrightarrow{N}_{Z})\leq \hat{n}\leq\mathcal{M}_U (\overrightarrow{N}_{Z})}{\min}\hat{n}}  + \frac{ \delta_1 + \gamma^{\epsATb}_{\mathrm{bin}}(\mathcal{M}_L (\overrightarrow{N}_{Z}), \deltaone)}{1-\delta_2-\gamma^{\epsATc}_{\mathrm{bin}}(\mathcal{M}_L (\overrightarrow{N}_{Z}), \delta_2)} \bigg) \\
    &\underset{(1)}{\leq} \Prdep \bigg({e_{\mathrm{ph},1}} > \frac{ \Findepdecoy \Big(\overrightarrow{N}_{X},  \underset{{\mathcal{M}_L (\overrightarrow{N}_{Z})\leq \hat{n}\leq\mathcal{M}_U (\overrightarrow{N}_{Z})}} {\max}\frac{\hat n}{1-\delta_2-\gamma^{\epsATc}_{\mathrm{bin}}(\hat n, \delta_2)}\Big)}{\underset{{\mathcal{M}_L (\overrightarrow{N}_{Z})\leq \hat{n}\leq\mathcal{M}_U (\overrightarrow{N}_{Z})}}{\min}\hat{n}}  + \frac{ \delta_1 + \gamma^{\epsATb}_{\mathrm{bin}}(\mathcal{M}_L (\overrightarrow{N}_{Z}), \deltaone)}{1-\delta_2-\gamma^{\epsATc}_{\mathrm{bin}}(\mathcal{M}_L (\overrightarrow{N}_{Z}), \delta_2)} \bigg) \\
    &=\Prdep \bigg({e_{\mathrm{ph},1}} > \frac{ \underset{{\mathcal{M}_L (\overrightarrow{N}_{Z})\leq \hat{n}\leq\mathcal{M}_U (\overrightarrow{N}_{Z})}}{\max} \Findepdecoy \Big(\overrightarrow{N}_{X},  \frac{\hat n}{1-\delta_2-\gamma^{\epsATc}_{\mathrm{bin}}(\hat n, \delta_2)}\Big)}{\underset{\mathcal{M}_L (\overrightarrow{N}_{Z})\leq \hat{n}\leq\mathcal{M}_U (\overrightarrow{N}_{Z})}{\min}\hat{n}}  + \frac{ \delta_1 + \gamma^{\epsATb}_{\mathrm{bin}}(\mathcal{M}_L (\overrightarrow{N}_{Z}), \deltaone)}{1-\delta_2-\gamma^{\epsATc}_{\mathrm{bin}}(\mathcal{M}_L (\overrightarrow{N}_{Z}), \delta_2)} \bigg) \\
    &\leq \Prdep \bigg({e_{\mathrm{ph},1}} > \max_{\mathcal{M}_L (\overrightarrow{N}_{Z})\leq \hat{n}\leq\mathcal{M}_U (\overrightarrow{N}_{Z})}\frac{  \Findepdecoy \Big(\overrightarrow{N}_{X}, \frac{\hat n}{1-\delta_2-\gamma^{\epsATc}_{\mathrm{bin}}(\hat n, \delta_2)}\Big)}{\hat{n}}  + \frac{ \delta_1 + \gamma^{\epsATb}_{\mathrm{bin}}(\mathcal{M}_L (\overrightarrow{N}_{Z}), \deltaone)}{1-\delta_2-\gamma^{\epsATc}_{\mathrm{bin}}(\mathcal{M}_L (\overrightarrow{N}_{Z}), \delta_2)} \bigg) \\
    &\underset{(2)}{\leq}\Prdep \bigg({e_{\mathrm{ph},1}} > \max_{\mathcal{M}_L (\overrightarrow{N}_{Z})\leq \hat{n}\leq\mathcal{M}_U (\overrightarrow{N}_{Z})}\frac{  \Findepdecoy (\overrightarrow{N}_{X}, \underset{N \in \mathcal{W}_{\delta_2}(\hat{n})}{\max} N)}{\hat{n}}  + \frac{ \delta_1 + \gamma^{\epsATb}_{\mathrm{bin}}(\mathcal{M}_L (\overrightarrow{N}_{Z}), \deltaone)}{1-\delta_2-\gamma^{\epsATc}_{\mathrm{bin}}(\mathcal{M}_L (\overrightarrow{N}_{Z}), \delta_2)} \bigg) \\
   &=\Prdep \bigg({e_{\mathrm{ph},1}} > \max_{\mathcal{M}_L (\overrightarrow{N}_{Z})\leq \hat{n}\leq\mathcal{M}_U (\overrightarrow{N}_{Z})}\frac{\underset{N \in \mathcal{W}_{\delta_2}(\hat{n})}{\max}   \Findepdecoy (\overrightarrow{N}_{X},N)}{\hat{n}}  + \frac{ \delta_1 + \gamma^{\epsATb}_{\mathrm{bin}}(\mathcal{M}_L (\overrightarrow{N}_{Z}), \deltaone)}{1-\delta_2-\gamma^{\epsATc}_{\mathrm{bin}}(\mathcal{M}_L (\overrightarrow{N}_{Z}), \delta_2)} \bigg)\\
   &\leq \epsspone^2 + \epssptwo^2 + \epssrc^2 + \epsATb^2 + \epsATc^2.
\end{aligned}
\end{equation}
The inequality (1) follows from
\begin{equation}
    \frac{\underset{{\mathcal{M}_L (\overrightarrow{N}_{Z})\leq \hat{n}\leq\mathcal{M}_U (\overrightarrow{N}_{Z})}}{\max} \hat n}{1-\delta_2-\underset{{\mathcal{M}_L (\overrightarrow{N}_{Z})\leq \hat{n}\leq\mathcal{M}_U (\overrightarrow{N}_{Z})}}{\min}\gamma^{\epsATc}_{\mathrm{bin}}(\hat n, \delta_2)}\geq \underset{{\mathcal{M}_L (N_{Z})\leq \hat{n}\leq\mathcal{M}_U (N_{Z})}} {\max}\frac{\hat n}{1-\delta_2-\gamma^{\epsATc}_{\mathrm{bin}}(\hat n, \delta_2)},
\end{equation}
and the fact that $\Findepdecoy$ is non-decreasing with respect to its second argument. On the other hand, the inequality (2) stems from
\begin{equation}
    \frac{\hat n}{1-\delta_2-\gamma^{\epsATc}_{\mathrm{bin}}(\hat n, \delta_2)} \geq \underset{N \in \mathcal{W}_{\delta_2}(\hat{n})}{\max} N,
\end{equation}
which holds from the definition of $\mathcal{W}_{\delta_2}m$ and the same monotonicity argument on $\Findepdecoy$.
\end{proof}

\section{Parameters for phase-error rate estimation}\label{app:list}
Here we provide the definitions of the parameters required for the estimation of the phase-error rate, which were omitted from the main text for the sake of clarity and readability. For a more detailed description of the quantities introduced below, we refer the reader to \cite{Guille_framework}.

In the following, $p_{t|j}$ denotes the probability of assigning a certain tag $t\in\{{\rm TAR, REF}\}$ given an event $j$,
\begin{equation}
\begin{aligned}
p_{\mathrm{TAR} \mid 0_Z} & =\frac{2 p_{\mathrm{TAR}} c_1 q_{\mathrm{vir} 1}}{p_{Z_A} p_{X_B}\left(1+c_1 q_{\mathrm{vir} 1}\right)}, \quad p_{\mathrm{TAR} \mid \mathrm{vir}}&=\frac{p_{\mathrm{TAR}}}{p_{Z_A} p_{Z_B}\left(1+c_1 q_{\mathrm{vir} 1}\right)}\\
p_{\mathrm{REF} \mid 0_X} & =\frac{2 p_{\mathrm{TAR}}\left(1-q_{\mathrm{vir} 1}\right)}{p_{X_A} p_{X_B}\left(1+c_1 q_{\mathrm{vir} 1}\right)}, \quad
p_{\mathrm{REF} \mid 1_X} & =\frac{2 p_{\mathrm{TAR}} c_3 q_{\mathrm{vir} 1}}{p_{X_A} p_{X_B}\left(1+c_1 q_{\mathrm{vir} 1}\right)},\\
p_{\mathrm{REF} \mid 1_Z} & =\frac{2 p_{\mathrm{TAR}} c_2 q_{\mathrm{vir} 1}}{p_{Z_A} p_{X_B}\left(1+c_1 q_{\mathrm{vir} 1}\right)},
\end{aligned} 
\end{equation}
where  $q_{\mathrm{vir} 1}$ is defined in \cref{eq:qvir} and $c_1$, $c_2$ and $c_3$ are defined in \cref{eq:c_defs}. These probabilities are needed to apply the Chernoff bounds. For that matter, we define the following means 
\begin{equation}
\begin{aligned}
& \mu_{0_Z, \mathrm{TAR}, Z_C,1}^{0_X}=p_{Z_C}  p_{\mathrm{TAR} \mid 0_Z}  \underline{N}_{0_Z,1}^{0_X}, \\
& \mu_{0_Z, \mathrm{TAR}, Z_C,1}^{\mathrm{det}}=p_{Z_C}  p_{\mathrm{TAR} \mid 0_Z}  \overline{N}_{0_Z,1}^{\mathrm{det}}, \\
& \mu_{\mathrm{key}, \mathrm{TAR}, Z_C,1}^{\mathrm{det}}=p_{Z_C}  p_{\mathrm{TAR} \mid \mathrm{vir}}  \hat{N}_{Z,1}^{\mathrm{det},Z}, \\
& \mu_{1_Z, \mathrm{REF}, Z_C,1}^{0_X}=p_{Z_C}  p_{\mathrm{REF} \mid 1_Z}  \overline{N}_{1_Z,1}^{0_X}, \\
& \mu_{0_X, \mathrm{REF}, Z_C,1}^{1_X}=p_{Z_C}  p_{\mathrm{REF} \mid 0_X}  \overline{N}_{0_X,1}^{1_X}, \\
& \mu_{1_X, \mathrm{REF}, Z_C,1}^{0_X}=p_{Z_C}  p_{\mathrm{REF} \mid 1_X}  \overline{N}_{1_X,1}^{0_X}, \\
& \mu_{1_Z, \mathrm{REF}, Z_C,1}^{\mathrm{det}}=p_{Z_C}  p_{\mathrm{REF} \mid 1_Z}  \underline{N}_{1_Z,1}^{\mathrm{det}}, \\
& \mu_{0_X, \mathrm{REF}, Z_C,1}^{\mathrm{det}}=p_{Z_C}  p_{\mathrm{REF} \mid 0_X}  \underline{N}_{0_X,1}^{\mathrm{det}}, \\
& \mu_{1_X, \mathrm{REF}, Z_C,1}^{\mathrm{det}}=p_{Z_C}  p_{\mathrm{REF} \mid 1_X}  \underline{N}_{1_X,1}^{\mathrm{det}}.
\end{aligned}
\end{equation}
With these, we can now compute both the lower bounds,
\begin{equation}
\begin{aligned}
& \underline{N}_{0_Z, \mathrm{TAR}, Z_C,1}^{0_X}=\max \left(\mu_{0_Z, \mathrm{TAR}, Z_C,1}^{0_X}-\sqrt{2 \mu_{0_Z, \mathrm{TAR}, Z_C,1}^{0_X} \ln \left(1 / \varepsilon_{\text {conc}}\right)},0\right), \\
& \underline{N}_{1_Z, \mathrm{REF}, Z_C,1}^{\mathrm{det}}=\max \left(\mu_{1_Z, \mathrm{REF}, Z_C,1}^{\mathrm{det}}-\sqrt{2 \mu_{1_Z, \mathrm{REF}, Z_C,1}^{\mathrm{det}} \ln \left(1 / \varepsilon_{\text {conc}}\right)},0\right), \\
& \underline{N}_{0_X, \mathrm{REF}, Z_C,1}^{\mathrm{det}}=\max \left(\mu_{0_X, \mathrm{REF}, Z_C,1}^{\mathrm{det}}-\sqrt{2 \mu_{0_X, \mathrm{REF}, Z_C,1}^{\mathrm{det}} \ln \left(1 / \varepsilon_{\text {conc}}\right)},0\right), \\
& \underline{N}_{1_X, \mathrm{REF}, Z_C,1}^{\mathrm{det}}=\max \left(\mu_{1_X, \mathrm{REF}, Z_C,1}^{\mathrm{det}}-\sqrt{2 \mu_{1_X, \mathrm{REF}, Z_C,1}^{\mathrm{det}} \ln \left(1 / \varepsilon_{\text {conc}}\right)},0\right),
\end{aligned}
\end{equation}
and the upper bounds,
\begin{equation}
\begin{aligned}
& \overline{N}_{\mathrm{key}, \mathrm{TAR}, Z_C,1}^{\mathrm{det}}=\mu_{\mathrm{key}, \mathrm{TAR}, Z_C,1}^{\mathrm{det}}+\sqrt{3 \mu_{\mathrm{key}, \mathrm{TAR}, Z_C,1}^{\mathrm{det}} \ln \left(1 / \varepsilon_{\mathrm{conc}}\right)}, \\
& \overline{N}_{0_Z, \mathrm{TAR}, Z_C,1}^{\mathrm{det}}=\mu_{0_Z, \mathrm{TAR}, Z_C,1}^{\mathrm{det}}+\sqrt{3 \mu_{0_Z, \mathrm{TAR}, Z_C,1}^{\mathrm{det}} \ln \left(1 / \varepsilon_{\mathrm{conc}}\right)},\\
& \overline{N}_{1_Z, \mathrm{REF}, Z_C,1}^{0_X}=\mu_{1_Z, \mathrm{REF}, Z_C,1}^{0_X}+\sqrt{3 \mu_{1_Z, \mathrm{REF}, Z_C,1}^{0_X} \ln \left(1 / \varepsilon_{\mathrm{conc}}\right)}, \\
& \overline{N}_{0_X, \mathrm{REF}, Z_C,1}^{1 X}=\mu_{0_X, \mathrm{REF}, Z_C,1}^{1_X}+\sqrt{3 \mu_{0_X, \mathrm{REF}, Z_C,1}^{1_X} \ln \left(1 / \varepsilon_{\mathrm{conc}}\right)}, \\
& \overline{N}_{1_X, \mathrm{REF}, Z_C,1}^{0_X}=\mu_{1_X, \mathrm{REF}, Z_C,1}^{0_X}+\sqrt{3 \mu_{1_X, \mathrm{REF}, Z_C,1}^{0_X} \ln \left(1 / \varepsilon_{\mathrm{conc}}\right)}.
\end{aligned}
\end{equation}
Also we define the aggregate quantities required to evaluate for instance \cref{eq:number_errors}, which are given by
\begin{equation}
\begin{aligned}
& \overline{N}_{Z_C=1,1}^{\mathrm{err}}=\overline{N}_{1_Z, \mathrm{REF}, Z_C,1}^{0_X}+\overline{N}_{0_X, \mathrm{REF}, Z_C,1}^{1_X}+\overline{N}_{1_X, \mathrm{REF}, Z_C,1}^{0_X}, \\
& \overline{N}_{Z_C=0,1}^{\mathrm{det}}=\overline{N}_{\mathrm{key}, \mathrm{TAR}, Z_C,1}^{\mathrm{det}}+\overline{N}_{0_Z, \mathrm{TAR}, Z_C,1}^{\mathrm{det}}, \\
& \underline{N}_{Z_C=1,1}^{\mathrm{det}}=\underline{N}_{1_Z, \mathrm{REF}, Z_C,1}^{\mathrm{det}}+\underline{N}_{0_X, \mathrm{REF}, Z_C,1}^{\mathrm{det}}+\underline{N}_{1_X, \mathrm{REF}, Z_C,1}^{\mathrm{det}} .
\end{aligned}
\end{equation}
The only remaining quantity to evaluate \cref{eq:number_errors} is $\overline{N}_{X_C=1,1}$, defined as
\begin{equation}
\overline{N}_{X_C=1,1}=N p_1p_{\text {TAR }}  p_{X_C}  \Delta_{\text {coin }}+\sqrt{3 N  p_1  p_{\text {TAR }}  p_{X_C}  \Delta_{\text {coin }}  \ln \left(1 / \varepsilon_{\text {conc }}\right)} .
\end{equation}
Here, $p_1=\sum_pp_{\mu_p}p_{1|\mu_p}$ represents the probability that Alice generates a single-photon and $\Delta_{\text {coin }}$ represents an upper bound on the coin imbalance parameter \cite{Guille_framework,GLLP}
defined as $\Delta_{\text {coin }}:=1-\operatorname{Re}\left\langle\Psi_{\operatorname{REF}} \mid \Psi_{\operatorname{TAR}}\right\rangle$ which incorporates the effect of the side channels in the security proof.

As proven in \cite{Guille_framework}, an upper bound on such quantity can be found by solving the following semidefinite programming problem 
\begin{equation}
\begin{array}{ll}
\min _G \operatorname{Re}\left\langle\Psi_{\operatorname{REF}} \mid \Psi_{\operatorname{TAR}}\right\rangle& \\
\text { s.t. } G \succeq 0 , & \forall j \in\{1,2,3,4\} , \\
G[j, j]=\left\langle\phi_j \mid \phi_j\right\rangle_B=1, & \forall j \in\{1,2,3,4\} , \\
G[j+4, j+4]=\left\langle\phi_j^{\perp} \mid \phi_j^{\perp}\right\rangle_B=1, & \forall j \in\{1,2,3,4\} , \\
G[j, j+4]=\left\langle\phi_j \mid \phi_j^{\perp}\right\rangle_B=0, & \\
G[1,2]=\left\langle\phi_1 \mid \phi_2\right\rangle_B=\cos (\kappa \pi / 2) , & \\
G[1,3]=\left\langle\phi_1 \mid \phi_3\right\rangle_B=\cos (\kappa \pi / 4) , & \\
G[1,4]=\left\langle\phi_1 \mid \phi_4\right\rangle_B=\cos (3 \kappa \pi / 4) , & \\
G[2,3]=\left\langle\phi_2 \mid \phi_3\right\rangle_B=\cos (\kappa \pi / 4) , & \\
G[2,4]=\left\langle\phi_2 \mid \phi_4\right\rangle_B=\cos (\kappa \pi / 4) , & \\
G[3,4]=\left\langle\phi_3 \mid \phi_4\right\rangle_B=\cos (\kappa \pi / 2) , &
\end{array}
\end{equation}
where for simplicity we have redefined $\left\{0_Z, 1_Z, 0_X, 1_X\right\}$ as $\{1,2,3,4\}$ and we have introduced the state $\ket{\phi_j^{\perp}}_B$, which represents any orthogonal state to $\ket{\phi_j}_B$.  The objective function is given by
\begin{equation}
\begin{aligned}
\operatorname{Re}\left\langle\Psi_{\operatorname{REF}} \mid \Psi_{\operatorname{TAR}}\right\rangle & =\frac{1}{1+q_{\mathrm{vir} 1} c_1}\left[\frac{\sqrt{1-q_{\mathrm{vir} 1}}}{2}\left(G^{\prime}[3,1]+G^{\prime}[3,2]\right)+\sqrt{q_{\mathrm{vir} 1} c_2}\left(a \sqrt{q_{\mathrm{vir} 1} c_1}+\frac{b}{2}\right) G^{\prime}[2,1]\right. \\
& \left.-\frac{b \sqrt{q_{\mathrm{vir} 1} c_2}}{2}+\sqrt{q_{\mathrm{vir} 1} c_3}\left(b \sqrt{q_{\mathrm{vir} 1} c_1}-\frac{a}{2}\right) G^{\prime}[4,1]+\frac{a \sqrt{q_{\mathrm{vir} 1} c_3}}{2} G^{\prime}[4,2]\right],
\end{aligned}
\end{equation}
where the parameters $a= \sqrt{\frac{\cos (\kappa \pi)-\cos (\kappa \pi / 2)}{\cos (\kappa \pi / 2)-1}},$ $b= -\sqrt{-2 \cos (\kappa \pi / 2)}$  and
\begin{equation}
G^{\prime}[i, j]=(1-\epsilon) \operatorname{Re} G[i, j]+\sqrt{(1-\epsilon) \epsilon} \operatorname{Re} G[i, j+4]+\sqrt{\epsilon(1-\epsilon)} \operatorname{Re} G[i+4, j]+\epsilon \operatorname{Re} G[i+4, j+4].
\end{equation}
Here, the quantity $\epsilon = \max_{j_k} \epsilon_{j_k}$ where the terms $\epsilon_{j_k}$ correspond to the conditions
\begin{equation}
\label{eq:general_close}
\left|\langle{\phi}_{j_{k-l_c}^{k+l_c}}^{{m_k=1},k}|_{T_k} \otimes\langle\lambda_{j_{k-l_c+1}^{k-1}, j_{k+1}^{k+l_c}}^{k}|_{I_{k+1}^{k+l_c}M_{k+1}^{k+l_c}T_{k+1}^{k+l_c}} \mid \Psi_{j_{k-l_c}^{k+l_c}}^{m_{k}=1,k}\rangle_{I_{k+1}^{k+l_c}M_{k+1}^{k+l_c}T_k^{k+l_c}}\right|^2 \geq 1-\epsilon_{j_k}.
\end{equation}
These are the conditions needed to apply \cref{coro:2}, and employ the protocol partition strategy. By using the assumptions given by \cref{eq:corre_WCP,eq:close_2,eq:close} from the main text, we obtain that
\begin{equation}
\label{eq:subsection}
\epsilon\leq \sum_{l=1}^{l_c}\epsilon_{l}^{\rm WCP}+\left(\sqrt{\epsilon_{\rm side}}+\sum_{l=1}^{l_c}\sqrt{\epsilon_{l}^{\rm SP}}\right)^2.
\end{equation}
The derivation of this bound is presented in the following section.
\subsection*{Derivation of Eq.~\eqref{eq:subsection}}

First, note that the global purification factorises as
\begin{equation}\label{eq:product}
    |\Psi^{m_k=1,k}_{j^{k+l_c}_{k-l_c}}\rangle_{I^{k+l_c}_{k+1}M^{k+l_c}_{k+1}T^{k+l_c}_k}
    \;=\;
    |\psi^{m_k=1,k}_{j^k_{k-l_c}}\rangle_{T_k}
    \;\otimes\;
    |\Psi_{j^{k+l_c}_{k-l_c+1}}\rangle_{I^{k+l_c}_{k+1}M^{k+l_c}_{k+1}T^{k+l_c}_{k+1}},
\end{equation}
where $|\psi^{m_k=1,k}_{j^k_{k-l_c}}\rangle_{T_k}$ is the single-photon component emitted in round~$k$,
and $|\Psi_{j^{k+l_c}_{k-l_c+1}}\rangle_{I^{k+l_c}_{k+1}M^{k+l_c}_{k+1}T^{k+l_c}_{k+1}}$ represents the purification of the
subsequent pulses.

Substituting Eq.~\eqref{eq:product} into Eq.~\eqref{eq:general_close}  it follows that
\begin{equation}
    \left|\langle\phi^{m_k=1,k}_{j^{k+l_c}_{k-l_c}}|\otimes
    \langle\lambda_{j_{k-l_c+1}^{k-1}, j_{k+1}^{k+l_c}}^{k}|\cdot|\Psi^{m_k=1,k}_{j^{k+l_c}_{k-l_c}}\rangle\right|^2
    \;=\;
    \underbrace{\left|\langle\phi^{m_k=1,k}_{j^{k+l_c}_{k-l_c}}|\psi^{m_k=1,k}_{j^k_{k-l_c}}\rangle
    \right|^2}_{\;\geq\;1-\epsilon_{\mathrm{ideal}}}
    \;\times\;
    \underbrace{\left|\langle\lambda^{k}_{j^{k-1}_{k-l_c+1},j^{k+l_c}_{k+1}}|
    \Psi_{j^{k+l_c}_{k-l_c+1}}\rangle\right|^2}_{\;\geq\;1-\epsilon_{\mathrm{correlations}}},
\end{equation}
where we have removed the indexes of the subsystems for simplicity of notation. Here, $\epsilon_{\mathrm{ideal}}$ quantifies how far $|\psi^{m_k=1,k}_{j^k_{k-l_c}}\rangle$ is from the reference state, while $\epsilon_{\mathrm{correlations}}$ quantifies the information leakage through subsequent pulses. Using $(1-a)(1-b) \geq 1 - a - b$ for $a,b \in [0,1]$, we obtain
\begin{equation}\label{eq:decomp}
    \epsilon_{j_k} \;\leq\; \epsilon_{\mathrm{ideal}} + \epsilon_{\mathrm{correlations}}.
\end{equation}

Now, we use the relation between the fidelity and trace distance, $i.e.$,
$T(|\phi\rangle,|\psi\rangle) = \sqrt{1 - |\langle\phi|\psi\rangle|^2}$,
together with the triangle inequality for the trace distance.
First, let us look at Eq.~\eqref{eq:close_2}. This equation indicates that a distinct setting at round $k-l$  changes the single-photon state, in trace distance, by at most $\sqrt{\epsilon_l^{\mathrm{SP}}}$, that is
\begin{equation}
    T\!\left(|\psi^{m_k=1,k}_{j^k_{k-l_c}}\rangle,\;
    |\psi^{m_k=1,k}_{j^{k-l-1}_{k-l_c},\tilde{j}_{k-l},j^k_{k-l+1}}\rangle\right)
    \;\leq\; \sqrt{\epsilon_l^{\mathrm{SP}}}.
\end{equation}
Applying the triangle inequality across all $l_c$ setting differences
that separate $j^{k-1}_{k-l_c}$ from the fixed reference sequence
$\hat{j}^{k-1}_{k-l_c}:=\{j_{k-1}=\gamma,...,j_{k-l_c}=\gamma\}$, we accumulate at most $l_c$ steps, meaning
\begin{equation}\label{eq:tri1}
    T\!\left(|\psi^{1}_{j^k_{k-l_c}}\rangle,\;
    |\psi^{1}_{\hat{j}^{k-1}_{k-l_c},j_k}\rangle\right)
    \;\leq\; \sum_{l=1}^{l_c}\sqrt{\epsilon_l^{\mathrm{SP}}}.
\end{equation}
Second, Eq.~\eqref{eq:close} guarantees that
the actual emitted states with fixed past settings are close to the ideal reference state, which in trace distance reads
\begin{equation}\label{eq:tri2}
    T\!\left(|\phi^{1}_{j_k}\rangle,\;
    |\psi^{1}_{\hat{j}^{k-1}_{k-l_c},j_k}\rangle\right)
    \;\leq\;\sqrt{\epsilon_{\mathrm{side}}}.
\end{equation}
Applying the triangle inequality to Eqs.~\eqref{eq:tri1}--\eqref{eq:tri2}
and converting back to fidelity we obtain
\begin{equation}\label{eq:ideal_bound}
    \epsilon_{\mathrm{ideal}}
    \;\leq\;
    \left(\sqrt{\epsilon_{\mathrm{side}}} + \sum_{l=1}^{l_c}\sqrt{\epsilon_l^{\mathrm{SP}}}\right)^{\!2}.
\end{equation}

Next, we define the reference state for future rounds by
fixing the current-round setting to a particular value $\gamma$, that is
\begin{equation}
    |\lambda^{k}_{j^{k-1}_{k-l_c+1},j^{k+l_c}_{k+1}}\rangle
    \;:=\;
    |{\Psi}_{j^{k-1}_{k-l_c+1},\gamma,j^{k+l_c}_{k+1}}\rangle.
\end{equation}
Eq.~\eqref{eq:corre_WCP} guarantees that, for each pulse separation~$l$ we have that,
\begin{equation}
    \left|\langle \Psi_{j^{k-l-1}_{k-l_c+1},\gamma,j^k_{k-l+1}} |
    \Psi_{j^k_{k-l_c+1}}\rangle\right|^2
    \;\geq\; 1 - \epsilon_l^{\mathrm{WCP}}.
\end{equation}
The future-pulse subsystem involves $l_c$ such steps. Their combined effect satisfies
\begin{equation}\label{eq:corr_bound}
    \epsilon_{\mathrm{correlations}}
    \;\leq\; 1 - \prod_{l=1}^{l_c}(1-\epsilon_l^{\mathrm{WCP}})
    \;\leq\; \sum_{l=1}^{l_c}\epsilon_l^{\mathrm{WCP}},
\end{equation}
where in the last inequality we have applied the Weierstrass product inequality $i.e.,$
$\prod_l(1-a_l)\geq 1-\sum_l a_l$ for $a_l\in[0,1]$.

Finally, substituting Eqs.~\eqref{eq:ideal_bound} and~\eqref{eq:corr_bound} into Eq.~\eqref{eq:decomp} we have that,
\begin{equation}
    \epsilon_{j_k}
    \;\leq\;
    \sum_{l=1}^{l_c}\epsilon_l^{\mathrm{WCP}}
    \;+\;
    \left(\sqrt{\epsilon_{\mathrm{side}}} + \sum_{l=1}^{l_c}\sqrt{\epsilon_l^{\mathrm{SP}}}\right)^{\!2}.
\end{equation}
Since this bound holds for all $j_k$, taking
$\epsilon = \max_{j_k}\epsilon_{j_k}$ yields Eq.~\eqref{eq:subsection}.

\section{Decoy-state analysis for phase error-rate estimation}\label{app:LP}

As discussed in the main text, extending the framework of \cite{Guille_framework} to estimate the phase-error rate of decoy-state protocols requires to obtain upper and lower bounds on $N_{j,1}^{{\rm det}, b_X}$ for different values of $j$ and $b$ via a decoy-state analysis. This can be accomplished by suitably modifying the linear program in \cref{eq:LP_key}, with the specific modifications detailed below.

First, instead of considering that Alice generates a $Z$-basis state, we now consider that she generates a specific BB84 state $j$. Furthermore, we now condition on Bob obtaining a certain outcome $b_X$ not on his basis choice. For that matter, let us define 
\begin{equation}
\delta_{\mu_p, N^{\rm{det},b_X}_{j}}^{b_X,\Delta}=-N^{{\rm det},b_X}_{j, \mu_p}+ N^{{\rm det},b_X}_{j} \mathcal{F}_{N^{{\rm det},b_X}_{j}, \varepsilon_{\mu_p, N^{{\rm det},b_X}_{j}}^{\Delta}}^{\Delta}\left(N^{{\rm det},b_X}_{j, \mu_p} / N^{{\rm det},b_X}_{j}\right),
\end{equation}
and  
\begin{equation}
\delta_{j, m}^{\Delta}=N_{j} \mathcal{G}_{N_{j}, \varepsilon_{j, m}^{\Delta}}^{\Delta}\left(p_{m}\right)-p_{m} N_{j} \quad\text{for}\quad \Delta\in\{\rm L,U\}.
\end{equation}

Finally, in the second constraint on the LP provided by \cref{eq:LP_key}, we substitute $N_{Z}^{Z}$, by $N_j$ removing the condition on Bob basis choice.

By implementing these changes, the linear program yields,
\begin{equation}
\begin{aligned}\label{eq:LP_beta}
\underline{N}_{j,1}^{{\rm det},b_X}=\text { min } & N^{{\rm det},b_X}_{j,1} \\
\text { s.t. } & N^{{\rm det},b_X}_{j, \mu_p}+\delta^{b_X}_{\mu_p, N^{{\rm det},b_X}_j} \geq \sum_{m=0}^{m_{\rm max}} p_{\mu_p \mid m} N^{{\rm det},b_X}_{j, m} \\
& N^{{\rm det},b_X}_{j, \mu_p}+\delta^{b_X}_{\mu_p, N^{{\rm det},b_X}_j} \leq \sum_{m=0}^{m_{\rm max}} p_{\mu_p \mid m} N^{{\rm det},b_X}_{j, m}+N_j\left(1-\sum_{m=0}^{m_{\rm max}} p_{m}\right)+\delta_{j,> m_{\rm max}} \\
& 0 \leq N^{{\rm det},b_X}_{j, m} \leq \min \left(p_{m} N_j+\delta_{j, m}, N^{{\rm det},b_X}_j\right) \\
& \delta_{\mu_p, N^{{\rm det},b_X}_j}^{b_X,\mathrm{L}} \leq \delta^{ b_X}_{\mu_p, N^{{\rm det},b_X}_j} \leq \delta_{\mu_p, N^{{\rm det},b_X}_j}^{b_X,\mathrm{U}} \\
& \delta_{j,> m_{\rm max}}^{\mathrm{L}} \leq \delta_{j,> m_{\rm max}} \leq \delta_{j,> m_{\rm max}}^{\mathrm{U}} \\
& \delta_{j, m}^{\mathrm{L}} \leq \delta_{j, m} \leq \delta_{j, m}^{\mathrm{U}} \\
& \sum_{p=1}^{\mathcal{P}} \delta^{b_X}_{\mu_p, N^{{\rm det},b_X}_j}=0 \\
& \delta_{j,\leq m_{\rm max}}^{ \mathrm{L}} \leq \sum_{m=0}^{m_{\rm max}} \delta_{j, m} \leq \delta_{j,\leq m_{\rm max}}^{\mathrm{U}} \\
& \sum_{m=0}^{m_{\rm max}} \delta_{j, m}+\delta_{j,> m_{\rm max}}=0 \\
& \forall~1 \leq p \leq \mathcal{P}, \forall~0 \leq m \leq m_{\rm max},
\end{aligned}
\end{equation}
which holds except with probability
\begin{equation}
\varepsilon_{j,1}=\sum_{m=0}^{m_{\rm max}}\left(\varepsilon_{j, m}^{\mathrm{L}}+\varepsilon_{j, m}^{\mathrm{U}}\right)+\sum_{p=1}^\mathcal{P}\left(\varepsilon_{\mu_p, N^{{\rm det},b_X}_j}^{\mathrm{L}}+\varepsilon_{\mu_p, N^{{\rm det},b_X}_j}^{\mathrm{U}}\right)+\varepsilon_{j,\leq m_{\rm max}}^{\mathrm{L}}+\varepsilon_{j,\leq m_{\rm max}}^{\mathrm{U}}+\varepsilon_{j,> m_{\rm max}}^{\mathrm{L}}+\varepsilon_{j,> m_{\rm max}}^{\mathrm{U}}.
\end{equation}

Note that the same linear program could yield an upper bound $\overline{N}_{j,1}^{{\rm det},b_X}$ by changing the minimization for a maximization.

\section{Channel model}\label{app:channel}

We generalize to the decoy-state setting the channel model  introduced in \cite{Guille_framework}. For that matter, let \(j\in\{0_Z,1_Z,0_X,1_X\}\) denote Alice's state choice, with associated angles $\theta_j \in \left\{0,\frac{\pi}{2},\frac{\pi}{4},\frac{3\pi}{4}\right\}$.
Alice emits a phase-randomized weak coherent pulse with average intensity \(\mu_p\in \{\mu_1,\mu_2,\mu_3\}\). To model intensity fluctuations, we assume that the actual intensity \(\mu'_{p}\) is uniformly distributed over the following interval
\begin{equation}
\label{eq:int_dist1}
[\mu_p(1-\delta_{\mu_p}),\,\mu_p(1+\delta_{\mu_p})],
\end{equation}
with density function
\begin{equation}
\label{eq:int_dist2}
g_{\mu_{p}}(\alpha)=\frac{1}{2\mu_p\delta_{\mu_p}},
\end{equation}
even though the proof could straightforwardly accommodate other intensity distribution models.

On Bob's side, he actively chooses a basis \(\xi\in\{Z,X\}\). It is convenient to encode the corresponding basis angle as $\vartheta_\xi=\frac{\pi}{4}\delta_{\xi,X}$, where $\delta_{\xi,X}$ is a Kronecker delta, so that \(\vartheta_Z=0\) and \(\vartheta_X=\pi/4\). We also include a polarization misalignment \(\theta_{\mathrm{mis}}\), as well as a factor $\kappa=1+\delta_{\rm SPF}/\pi$ that serves to model the state preparation flaws for certain parameter $\delta_{\rm SPF}$, and define the effective signal angle at Bob as
\begin{equation}
\tilde{\theta}_j=\kappa\theta_j+\theta_{\mathrm{mis}}.
\end{equation}
This means that, for a fixed basis \(\xi\), the fraction of light directed to Bob's detector associated with the bit \(0\) is $\cos^2(\tilde{\theta}_j-\vartheta_\xi)$, while the orthogonal output, associated with the bit \(1\), carries weight $\sin^2(\tilde{\theta}_j-\vartheta_\xi)$.
Since exchanging the two output ports is equivalent to a shift by \(\pi/2\), both cases can be written compactly by introducing the effective angle
\begin{equation}
\phi_{j,b,\xi}
=
\tilde{\theta}_j-\vartheta_\xi+\frac{\pi}{2}b
=
\kappa\theta_j+\theta_{\mathrm{mis}}-\frac{\pi}{4}\delta_{\xi,X}+\frac{\pi}{2}b,
\end{equation}
which is defined for Bob's outcomes $b\in\{0,1\}$ (in general, $b\in\{0,1,\perp\}$, with $\perp$ being the no detection event). Note that \(\cos^2(\phi_{j,b,\xi})\) gives the fraction of light impinging on detector \(b\), whereas \(\sin^2(\phi_{j,b,\xi})\) gives the fraction impinging on the orthogonal detector.

The total transmittance \(\eta\) is taken to include both channel loss and detector efficiency. For a given actual intensity \(\mu'_p\), the no-click probabilities of detector $b\in\{0,1\}$ and of the orthogonal detector are
\begin{equation}
r_b(j,\mu'_p,\xi)=(1-d_{\rm dc})e^{-\eta\mu'_p\cos^2(\phi_{j,b,\xi})},
\end{equation}
and
\begin{equation}
\bar r_b(j,\mu'_p,\xi)=(1-d_{\rm dc})e^{-\eta\mu'_p\sin^2(\phi_{j,b,\xi})},
\end{equation}
respectively, where \(d_{\rm dc}\) is the dark-count probability of each detector. Thus, assuming threshold detectors and random assignment of double-click events to single-click events, the conditional probability of Bob obtaining the conclusive outcome \(b\in\{0,1\}\) is
\begin{equation}
\begin{split}
P(b\mid j,\mu'_p,\xi)
&=
\Big(1-r_b(j,\mu'_p,\xi)\Big)\bar r_b(j,\mu'_p,\xi)
+\frac{1}{2}\Big(1-r_b(j,\mu'_p,\xi)\Big)\Big(1-\bar r_b(j,\mu'_p,\xi)\Big)
\\
& =
\left[1-(1-d_{\rm dc})e^{-\eta\mu'_p\cos^2(\phi_{j,b,\xi})}\right]
\left[
(1-d_{\rm dc})e^{-\eta\mu'_p\sin^2(\phi_{j,b,\xi})}
+
\frac{1}{2}\left(1-(1-d_{\rm dc})e^{-\eta\mu'_p\sin^2(\phi_{j,b,\xi})}\right)
\right],
\end{split}
\end{equation}
and the probability of an inconclusive event $b=\perp$ is
\begin{equation}
P(\perp\mid j,\mu'_p,\xi)=(1-d_{\rm dc})^2e^{-\eta\mu'_p}.
\end{equation}
Finally, averaging over the intensity fluctuations yields
\begin{equation}
P(b\mid j,\mu_p,\xi)
=
\int_{\mu_p(1-\delta_{\mu_p})}^{\mu_p(1+\delta_{\mu_p})}
g_{\mu_p}(\alpha)\,P(b\mid j,\alpha,\xi)\,d\alpha,
\qquad
b\in\{0,1,\perp\}.
\end{equation}

Lastly, note that the statistics of this channel model are independent of the random global phase of the optical pulses, and therefore for simplicity we have considered a fixed phase.

\section{Numerical simulations}\label{app:numerical_sim}

Here, we present additional numerical simulations to assess the performance of decoy-state QKD as a function of the side-channel leakage $\epsilon_{\rm side}$, state-preparation error $\delta_{\rm SPF}$, correlation strength $\epsilon_{1}$, detector tolerances $\Delta_{\rm det}=\Delta_{\rm dc}$, and total number of rounds $N$. For that matter, we employ the channel model described in \cref{app:channel} together with the parameter values listed in \cref{tab:values}. 

In all the figures in this Appendix, like those in the main text, we plot the finite-key secret key rate as a function of the channel loss, not the overall system transmittance. In particular, the influence of unbounded bit and basis correlations is analysed in \cref{fig:sweep_epsilonl}, where we evaluate the secret-key rate for different values of the correlations strength $\epsilon_1$.  \cref{fig:sweep_deltas} investigates the joint effect of state-preparation flaws $\delta_{\rm SPF}$ and detector tolerances $\Delta_{\rm det}=\Delta_{\rm dc}$. Finally, \cref{fig:sweep_blocks} shows the dependence of the secret-key rate on the block size for different values of $N$. Unless stated otherwise, the simulations are performed with $\epsilon_{\rm side}=10^{-6}$, $\delta_{\rm SPF}=0.063$, $\epsilon_{1}=10^{-4}$, $\Delta_{\eta}=\Delta_{\rm dc}=0.02$, and $N=10^{12}$. Note that in \cref{fig:sweep_deltas}, the effect of state-preparation flaws is magnified by  the detector tolerances. A state-preparation flaw sends a small fraction of each signal to the wrong detector, and this effect is amplified by the misalignment of the channel; thus, even for ideal detectors, increasing $\theta_{\rm mis}$ further separates the solid and dashed curves. In addition, since the detectors are characterized only within finite tolerances, the eavesdropper may bias the detection-efficiency mismatch in favour of that detector, enhancing the one that registers the erroneous component while suppressing the one that registers the correct outcome. The detector mismatch therefore acts as a lever that amplifies the cost of the flaw, so that the two imperfections combine and the penalty associated with the source flaws grows with the detector tolerances.

\begin{figure}[]
\centering
\begin{subfigure}[t]{0.49\linewidth}
\centering
    \includegraphics[width=\linewidth]{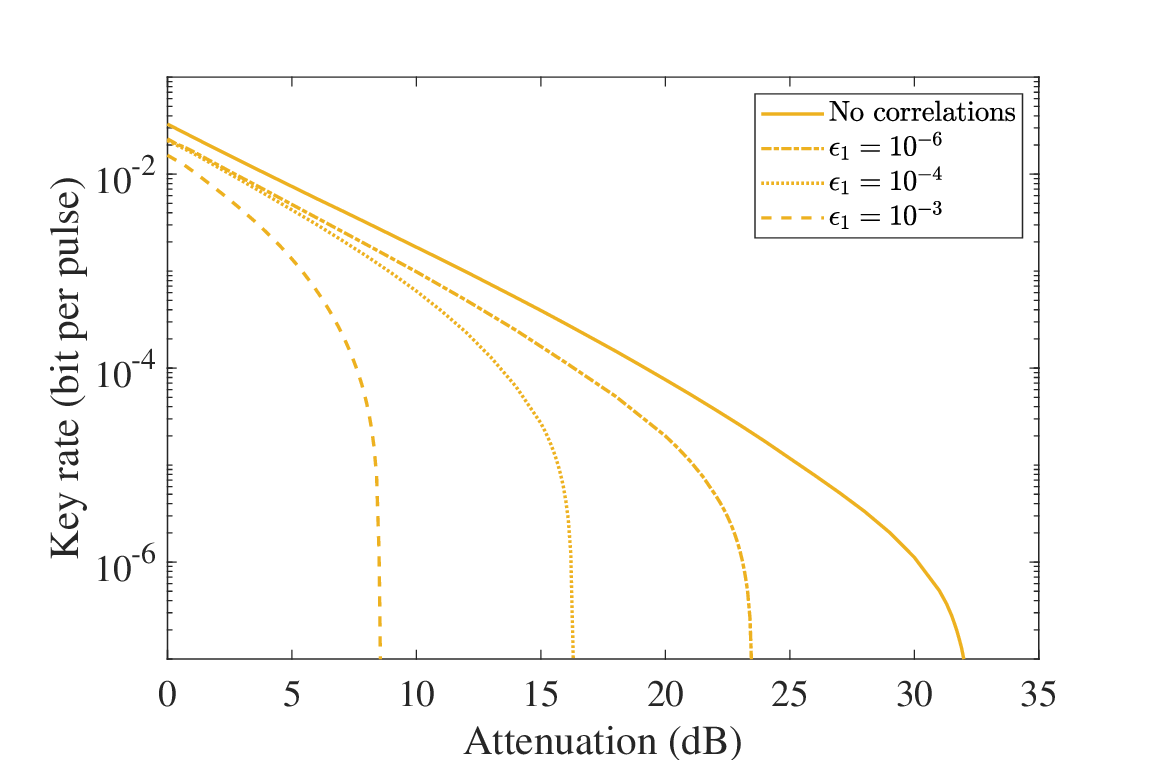}
    \caption{Finite secret-key rate of a decoy-state BB84 protocol incorporating both source and detector imperfections as a function of the channel loss in dB for different values of the correlation strength $\epsilon_1$. We set the state-preparation error parameter to $\delta_{\rm SPF}=0.063$~\cite{Guille_framework,SPF1,SPF2},  the side channel leakage to $\epsilon_{\rm side}=10^{-6}$, the detector tolerances to $\Delta_{\eta}=\Delta_{\rm dc}=0.02$ and we use a block size of $N=10^{12}$.}
    \label{fig:sweep_epsilonl}
\end{subfigure}
\hfill
\begin{subfigure}[t]{0.49\linewidth}
    \centering
    \includegraphics[width=\linewidth]{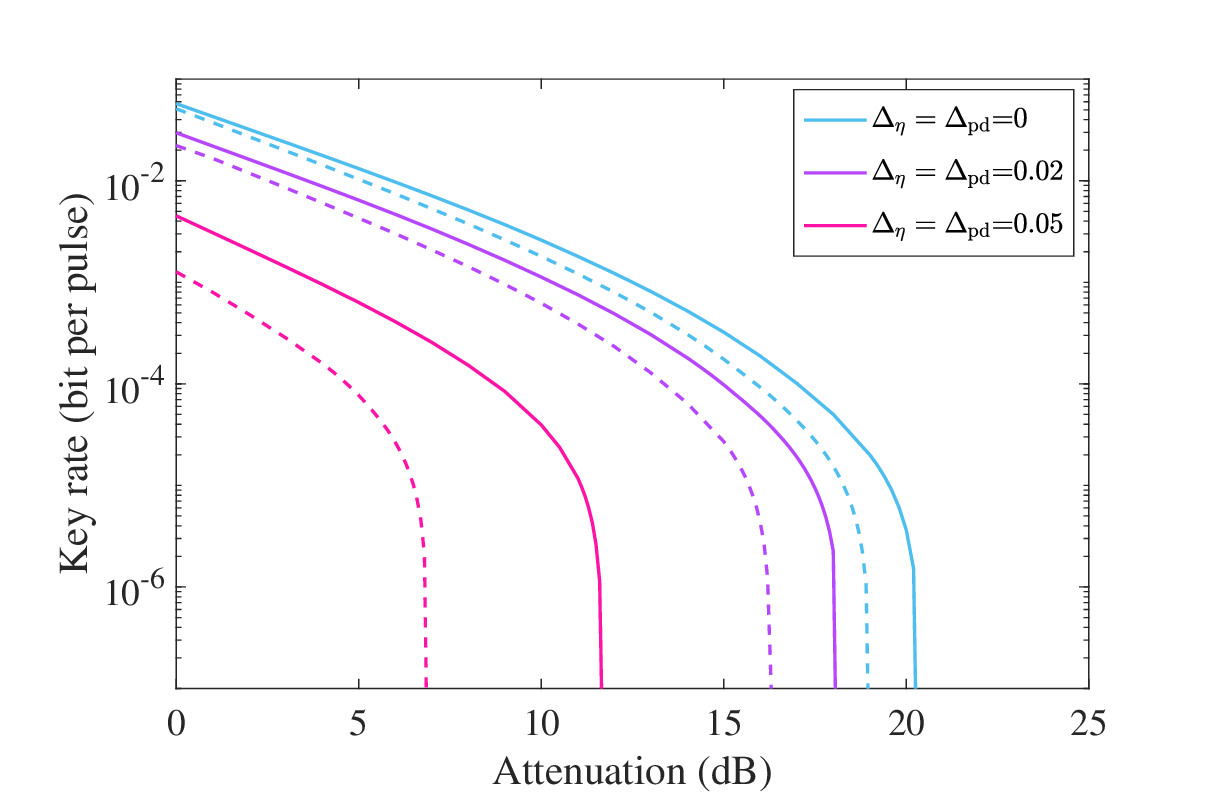}
    \caption{Finite secret-key rate of a decoy-state BB84 protocol incorporating both source and detector imperfections as a function of the channel loss in dB for different magnitudes of state preparation flaws $\delta_{\rm SPF}$ and detector tolerances $\Delta_{\eta}$ and $\Delta_{\rm dc}$. Solid lines represent $\delta_{\rm SPF}=0$, while dashed lines correspond to $\delta_{\rm SPF}=0.063$~\cite{Guille_framework,SPF1,SPF2}. We set the side channel leakage to $\epsilon_{\rm side}=10^{-6}$, the correlations strength to $\epsilon_1=10^{-4}$ and we use a block size of $N=10^{12}$.}
    \label{fig:sweep_deltas}
\end{subfigure}
\vspace{0.2cm}

\begin{subfigure}[t]{0.49\linewidth}
    \centering
    \includegraphics[width=\linewidth]{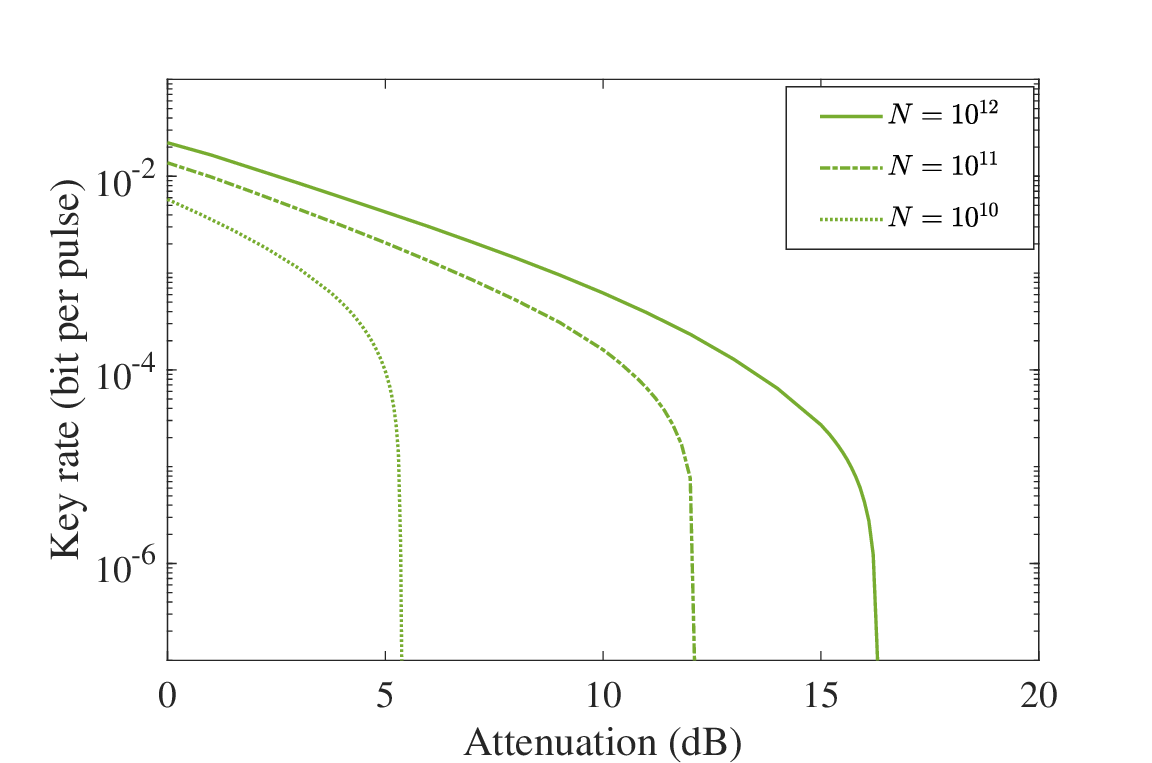}
    \caption{Finite secret-key rate of a decoy-state BB84 protocol incorporating both source and detector imperfections as a function of the channel loss in dB for different block sizes $N$. We set the state-preparation error parameter to $\delta_{\rm SPF}=0.063$~\cite{Guille_framework,SPF1,SPF2},  the side channel leakage to $\epsilon_{\rm side}=10^{-6}$, the correlations strength to $\epsilon_1=10^{-4}$ and the detector tolerances to $\Delta_{\eta}=\Delta_{\rm dc}=0.02$. }
    \label{fig:sweep_blocks}
\end{subfigure}

\caption{Numerical simulations of the finite secret-key rate of the decoy-state BB84 protocol in the presence of source and detector imperfections. The secret-key rate is shown as a function of the channel loss for different values of the block size $N$,  correlation strength $\epsilon_1$, state-preparation flaws $\delta_{\rm SPF}$ and detector tolerances $\Delta_{\eta}$ and $\Delta_{\rm dc}$.}
\label{fig:combined}
\end{figure}

\clearpage
\end{widetext}

\bibliography{refs}

\begin{thebibliography}{10}
\bibitem{a}
Note that in the protocol considered in the main text, Alice and Bob extract key from all detected rounds in which they select the $Z$ basis. For the sake of generality, in this protocol description we consider the possibility that only a random subset of the detected rounds in which $\protect \tilde {B}_k=\xi _k=\protect \mathrm {key}$ is used for sifted key generation.
\end{thebibliography}

\end{document}